\documentclass{article}

\usepackage{kr}

\usepackage[english]{babel}

\usepackage{times}
\usepackage{soul}
\usepackage{url}
\usepackage[hidelinks]{hyperref}
\usepackage[utf8]{inputenc}
\usepackage[small]{caption}
\usepackage{graphicx}
\usepackage{booktabs}
\usepackage{pifont}
\usepackage{listings}
\usepackage{todonotes}

\newif\ifArxivVersion

\ArxivVersiontrue

\usepackage[fig,nohypref,chgbar]{fmocdmac}

\usetikzlibrary{arrows,arrows.meta,calc}
\usetikzlibrary{shadows}
\usetikzlibrary{shapes.multipart}
\usetikzlibrary{positioning}
\usetikzlibrary{shadows}
\usetikzlibrary{shapes.multipart}
\usetikzlibrary{arrows.meta}
\usetikzlibrary{graphs,quotes}
\usepackage{xcolor-solarized}

\usepackage{cleveref}

\usepackage[normalem]{ulem}
\usepackage{xcolor}
\usepackage{soul}
\definecolor{myyellow}{RGB}{255,245,170}
\sethlcolor{myyellow}

\newtheorem{definition}{Definition}

\newtheorem{proposition}{Proposition}
\newtheorem{lemma}{Lemma}

\newtheorem{theorem}{Theorem}
\newtheorem{corollary}{Corollary}

\newtheorem{example}{Example}

\newcommand{\mcomment}[2]{\ifmmode\margincomment{#1}{#2}\else\footcomment{#1}{#2}\fi}
\newcommand{\footcomment}[2]{{\color{blue}\textbf{(#1)}}\footnote{\textbf{#1:} #2}}
\newcommand{\margincomment}[2]{{\color{blue}\textbf{(#1)}}\footnotemark\marginnote{\tiny\textsuperscript{\thefootnote}\textbf{#1:} #2}}

\newcommand{\filip}[1]{\mcomment{Filip}{#1}}

\usepackage[ruled,vlined,linesnumbered]{algorithm2e}
\SetKwComment{Comment}{/* }{ */}

\newcommand{\alogspace}{\mathsc{ALogSpace}}
\newcommand{\ptime}{\mathsc{P}}
\newcommand{\nptime}{\mathsc{NP}}

\newcommand{\graph}{\mathcal{G}}
\newcommand{\schema}{\mathcal{S}}
\newcommand{\cat}{\mathcal{C}}
\newcommand{\decl}[2]{#1 \colon\! #2}

\newcommand{\sel}{\mathit{sel}}

\newcommand{\nodes}{\mathsf{Nodes}}

\newcommand{\Names}{\mathsf{Names}}

\newcommand{\IRIs}{\mathsf{IRIs}}
\newcommand{\Blanks}{\mathsf{Blanks}}
\newcommand{\Literals}{\mathsf{Literals}}

\newcommand{\mathsc}[1]{{\normalfont\textsc{#1}}}

\newcommand{\OMIT}[1]{}

\newcommand{\pathExpr}{\pi}

\newcommand{\id}{\mathsf{id}}
\newcommand{\eq}{\mathsf{eq}}
\newcommand{\disj}{\mathsf{disj}}
\newcommand{\geqn}[2]{\exists^{\geq #1}#2.}
\newcommand{\leqn}[2]{\exists^{\leq #1}#2.}

\newcommand{\hasvalue}{\mathsf{test}}
\newcommand{\test}{\mathsf{test}}
\newcommand{\closed}{\mathsf{closed}}

\newcommand{\sem}[1]{\llbracket{#1}\rrbracket}
\newcommand{\semdelta}[3]{\llbracket{#1}\rrbracket^{#2}_{#3,\cat}}

\newcommand{\ValueTypes}{\mathcal{T}} 
\newcommand{\vtype}{\tau} 

\newcommand{\gDef}{{\color{orange} \ \Coloneqq \ }}
\newcommand{\gMid}{{\color{orange} \ \big|\ }}
\newcommand{\gEnd}{{\color{orange} \ .\ }}

\newcommand{\shexneigh}[1]{\big\{ #1 \big\}}
\newcommand{\shexneighzero}[1]{\shexneigh{#1 \shexeach \top}}

\newcommand{\shexref}{}
\newcommand{\shexeach}{\mathop{;}}
\newcommand{\shexone}{\mathop{|}}
\newcommand{\shexinverse}[1]{#1^{-}}
\newcommand{\shexneg}[1]{\neg #1}
\newcommand{\shexneginv}[1]{\neg{\shexinverse{#1}}}

\newcommand{\neigh}{\mathsf{Neigh}}

\newcommand{\const}{\mathsf{Const}}

\cmdmthsetext{Name}[\Names][s]
\usrmth{lsfp}{}{fun}
\usrmth{gsfp}{}{fun}

\cmdtxtoparname{All}
\cmdtxtoparname{SSL}
\cmdtxtoparname{SMS}
\cmdtxtoparname{LFP}
\cmdtxtoparname{GFP}

\newcommand{\microshex}{\ensuremath{\text{ShEx}_0}}
\newcommand{\microshacl}{\ensuremath{\text{SHACL}_0}}

\newcommand{\illustration}[1]{}

\newcommand{\testcode}[1]{\texttt{#1}}
\newcommand{\pass}{\color{RoyalBlue}}
\newcommand{\fail}{\color{BrickRed}}

\newcommand{\y}{{\pass {\Large $\bullet$}}}%

\newcommand{\n}{{\fail {\tiny $\blacksquare$}}}%

\newcommand{\cmark}{\pass \phantom{a}\ding{51}}%
\newcommand{\xmark}{\fail \phantom{a}\ding{55}\phantom{*}}%
\newcommand{\xxmark}{\fail \phantom{a}\ding{55}*}%

\newcommand{\ttc}[1]{\mathit{tc}(#1)}
\newcommand{\nbtc}[1]{\#(#1)}
\newcommand{\shapes}[1]{\mathit{shapes}(#1)}
\newcommand{\dettc}{\mathit{det}}
\newcommand{\preds}{\mathit{preds}}

\title{Common Foundations for Recursive Shape Languages}

\author{
Shqiponja Ahmetaj$^1$\!\!\and
Iovka Boneva$^2$\!\!\and
Jan Hidders$^3$\!\!\and
Maxime Jakubowski$^1$\!\!\and \\ 
Jose-Emilio Labra-Gayo$^4$\!\!\and
Wim Martens$^5$\!\!\and
Fabio Mogavero$^6$\!\!\and
Filip Murlak$^7$\!\!\and \\
Cem Okulmus$^8$\!\!\and
Ognjen Savkovi\'c$^9$\!\!\and
Mantas Šimkus$^1$\!\!\and
Dominik Tomaszuk$^{10}$ \\
\affiliations
$^1$TU Wien, Austria  \\
$^2$Univ. Lille, F-59000 Lille, France \\
$^3$Birkbeck, University of London, UK  \\
$^4$University of Oviedo, Spain  \\
$^5$University of Bayreuth, Germany  \\
$^6$Universita di Napoli Federico II, Italy  \\
$^7$University of Warsaw, Poland\\ 
$^8$Paderborn University, Germany  \\
$^9$Free University of Bolzano, Italy  \\
$^{10}$University of Bia{\l}ystok, Poland
}

\usepackage[dvipsnames]{xcolor}

\begin{document}

\maketitle


\begin{abstract}
As schema languages for RDF data become more mature, we are seeing efforts to extend them with recursive semantics, applying diverse ideas from logic programming and description logics.
While ShEx has an official recursive semantics based on  \emph{greatest fixpoints} (\GFP), the discussion for SHACL is ongoing and seems to be converging towards \emph{least fixpoints} (\LFP). 
A practical study we perform shows that, indeed, ShEx validators implement \GFP, whereas SHACL validators are more heterogeneous. 
This situation creates tension between ShEx and SHACL, as their semantic commitments appear to diverge, potentially undermining interoperability and predictability. 
We aim to clarify this design space by comparing the main semantic options in a principled yet accessible way, hoping to engage both theoreticians and practitioners, especially those involved in developing tools and standards.   
We present a unifying formal semantics that treats \LFP, \GFP, and supported model semantics (\SMS), clarifying their relationships and highlighting a duality between \LFP and \GFP on stratified fragments. 
Next, we investigate to which extent the directions taken by SHACL and ShEx are compatible. 
We show that, although ShEx and SHACL seem to be going in different directions, they include large fragments with identical expressive power. 
Moreover, there is a strong correspondence between these fragments through the aforementioned principle of duality.
Finally, we present a complete picture of the data and combined complexity of ShEx and SHACL validation under \LFP, \GFP, and \SMS, showing that \SMS comes at a higher computational cost under standard complexity-theoretic assumptions.
\end{abstract}






\section{Introduction}

Graph-based data representations are rapidly gaining importance due to the unprecedented growth of interconnected data and their increasing role in AI-driven systems~\cite{SakrBVIAAAABBDV21}. Furthermore, their flexibility makes them an ideal tool for bridging the gap between structured data and machine learning tools in industrial applications \cite{anuradha2023deep,mohamed2022rdfframes}. 
Since automatic processing of data (graph or otherwise) is significantly facilitated by the presence of a \emph{schema}, we have seen various efforts to design schema languages for graph-based data representations, notably ShEx~\cite{PGS14,SBG15,shex_standard}, 
SHACL~\cite{KK17}, and PG-Schema~\cite{ABDF21,ABDF23}. 
In this paper, we zoom in on \emph{recursion} as a foundational mechanism that directly shapes the semantics of cyclic schema dependencies. We therefore focus on ShEx and SHACL, the de facto standard schema languages for RDF, since PG-Schema does not use this kind of recursion.

It is well-known that designing declarative languages with recursion and negation is non-trivial~\cite{AHV95,Aref25}. SHACL and ShEx's open-world nature (contrasting the typical database setting) does not make this task any easier.
This might explain why the design of recursion is taking different directions in ShEx and SHACL. While ShEx has long had an official semantics based on \emph{greatest fixpoints} (\GFP) \cite{shex_standard,SBG15,BGP17},
the discussion for SHACL is ongoing.
Proposals include a \emph{supported-model semantics} (\SMS)~\cite{CRS18}, grounded in first-order logic, as well as  alternatives inspired by fixpoint and logic-programming paradigms~\cite{ACORSS20,OS24,BJ21,AhmetajLOS22,PKM22}. The academic discussion appears to be converging towards \emph{least fixpoints} (\LFP), but the issue remains to be addressed by the W3C working group.%

The current situation causes some tension between ShEx and SHACL since the proposals seem to diverge, which may cause interoperability issues in the future. Interoperability is important in large-scale industry applications. ShEx and SHACL can happily co-exist and have their own approach and identity, but the more \emph{fundamentally different} they are, the more painful it becomes to switch from one to the other if we realise at a late development stage that we have chosen the wrong language.  



\smallskip
\textit{Our Contributions.}
(1) We shed light on this matter by proposing a unifying formal framework that is compatible with the work of Ahmetaj et al.~\cite{www25}. We define a \emph{simple shape language} ($\SSL$) with the three main options for recursion: \LFP, \GFP, and \SMS
(Section~\ref{sec:semantics}). 
The language $\SSL$ is carefully designed to be much simpler than ShEx and SHACL, yet sufficiently powerful to illustrate the fundamental tradeoffs between \LFP, \GFP, and \SMS, and to test which semantics is used by validation engines. 
Despite its simplicity, it is still expressive: it corresponds to the alternation-free fragment of the modal $\mu$-calculus \cite{Koz83} with nominals but without atomic propositions, which is an important yardstick. Using the duality principle of fixpoint theory, we observe that $\SSL$ with \LFP is equally expressive as $\SSL$ with \GFP, using a straightforward translation. This close relationship is good news for the compatibility between ShEx with \GFP and SHACL with \LFP. The situation for \SMS 
differs, see 
(4).

\smallskip
(2)  
We experimentally explore whether existing SHACL and ShEx validators use \LFP, \GFP, or \SMS (Section~\ref{sec:thr;sub:comstdschval}). 
Our tests reveal that current tools implicitly commit to different semantics, 
which leads to non-uniform behaviour across implementations. In a nutshell, ShEx validators seem to implement \GFP, whereas SHACL validators are more heterogeneous. This study is an important motivation of our work. It confirms our hypothesis that validation engines are not aligned and 
underscores why there is a need for a consensus on recursive semantics in SHACL. Once a semantics for recursion in SHACL is agreed upon, our tests can be added to the official test suite~\cite{LKK24}.

\smallskip
(3) We then move to the formal study of real schema languages. Though neat and relatively expressive, $\SSL$ is still a small part of real ShEx and SHACL. In Sections~\ref{sec:full-shex-shacl}--\ref{sec:comparison-shex-shacl}, we discuss translations 
between large fragments of ShEx (with \GFP) and SHACL with \LFP. Our translations show that these 
larger fragments have the same expressive power and that the duality principle that we saw for $\SSL$ extends to real-world languages.

\smallskip
(4) We provide a complete overview of the data and combined complexity of ShEx, SHACL, and $\SSL$ with \LFP, \GFP, and \SMS (Section~\ref{ssec:complexity}). In some cases the complexity has been known already \cite{ACORSS20,SBG15,BJVdB24,BJ21}; we complete the picture by establishing tight bounds for the combined and data complexity for ShEx. The main conclusion from the complete set of results is that \LFP and \GFP allow polynomial-time data complexity for all languages, whereas \SMS immediately leads to intractability, even for very simple languages like $\SSL$. This means that \LFP and \GFP are much easier to deal with from an algorithmic perspective, and would be our recommendation for recursion in ShEx and/or SHACL.


\smallskip
(5) As a group of authors including contributors to the design of both ShEx and SHACL, we reach a common understanding of the consequences of the choice of the semantics for recursion, and clarify that it is fine for SHACL to use \LFP, while ShEx uses \GFP. {In a nutshell, \LFP is a very sensible choice because it is natural, its expressiveness is compatible with ShEx's \GFP (allowing effective translation), and has good complexity properties (see also Section~\ref{sec:conclusion}).

\ifArxivVersion
Full proofs can be found in the appendix. 
\else 
Missing proofs can be found in the full version of this paper~\cite{arxiv}.
\fi

\subsubsection*{Related Work.} 



The development of constraint languages for graph-based data has a rich history. SHACL~\cite{KK17} was designed taking the initial inspiration from SPIN Rules~\cite{KHI11}, whereas ShEx~\cite{PGS14} draws more heavily on paradigms established in XML Schema and RELAX NG~\cite{xsd,xml,MNNS17}. Both languages aim to enable \emph{prescriptive} validation of RDF, in contrast to the more \emph{descriptive} nature of languages such as OWL~\cite{OWL}.
Nevertheless, they diverge in several key features, which poses challenges for tool interoperability and for users seeking to transition between the two frameworks.
In our prior work~\cite{www25}, we take initial steps toward reconciling these differences by proposing a common framework for ShEx, SHACL, and PG-Schema, which was developed  to provide similar schema capabilities in the context of the recent GQL standard~\cite{GQL,gqlpaper}. The present work builds directly on that foundation, with a particular focus on recursion, a feature explicitly excluded from the earlier study.


Before this, expressiveness and computational properties of (non-recursive) SHACL have been extensively studied~\cite{BJB22,BJVdB24,PKMN20,PKM22}. Similarly, the foundations of ShEx, including its formal semantics and complexity, have been explored in depth~\cite{SBG15,BGP17}. %

\textit{Recursion in Shape Languages.}
The semantic challenges of recursion for shape languages, and particularly its interaction with negation, have been a significant topic of recent research. This is especially pronounced for SHACL, where the W3C recommendation~\cite{KK17} notably left the semantics of recursion undefined.  
The challenges of recursion were first formally addressed by Corman et al.~\cite{CRS18,CFRS19}, who proposed a \emph{supported model semantics} based on first-order logic.
Subsequent work has explored various semantics mirroring established paradigms from fixpoint logics and 
logic programming
~\cite{ACORSS20,OS24,BJ21,AhmetajLOS22,PKM22}. In contrast to SHACL, the semantics of recursion in ShEx has been consistently defined via \emph{greatest fixpoints}~\cite{BGP17,SBG15}. This 
divergence in the formalization process 
lead to 
significant inconsistencies across validators. 

\textit{Validation in Practice.} 
There is a growing body of work on the practical use of shape languages. This includes mining shapes from large knowledge graphs~\cite{RLH23,LissandriniRH24}, validating real-world datasets such as Wikidata~\cite{FSAP24}, and studying data provenance and explanations for validation outcomes~\cite{DDJB23,AhmetajDPS22}. The need to validate heterogeneous graph data has also motivated comparative studies of SHACL and ShEx~\cite{GPBK17,GGFE19,T17} as well as mappings between RDF and Property Graph paradigms~\cite{ATT20,H14}. Further extensions of ShEx for different graph types have been proposed~\cite{L24,L22}, alongside unifying graph data models such as OneGraph~\cite{LSHBBB23} and MillenniumDB’s Domain Graph Model~\cite{DCRMD23}. Semantic ambiguities in recursive SHACL directly affect these practical efforts, since the outcome of validation can depend on the choice of validator and its implicit semantics.

\textit{Connections to Description Logics and $\mu$-calculus.} Foundational work on SHACL has revealed strong connections to expressive Description Logics (DLs), the logics underlying OWL~\cite{BJB22,SRLS20,AhmetajDOPSS21}.
The problem of defining coherent semantics for recursive shape languages 
closely parallels the classic issue of \emph{terminological cycles} in DLs, extensively studied in the early 1990s. Terminologies with recursive definitions, where concept names may be defined directly or indirectly in terms of themselves, were recognized to require special semantic treatment to ensure meaningful interpretations. Initial studies~\cite{Baader1990,Baader1996,Nebel1991} identified three main approaches to interpreting such cycles: least fixpoint, greatest fixpoint, and descriptive (classical) semantics.
It was observed that, for many cyclic definitions, the greatest-fixpoint semantics yields more satisfactory results by allowing maximal solutions consistent with their intended meaning, whereas for others, the least-fixpoint semantics is more suitable
~\cite{Baader1996,DeGiacomoLenzerini1994}.
Further work argued that, instead of selecting a single semantics, better results are obtained by adopting a formalism that allows for the different semantics to coexist, where the correspondence with the modal~$\mu$-calculus was also established~\cite{Schild1994,DeGiacomoLenzerini1994}.
In our setting, we show that expressive yet tractable fragments of recursive SHACL and ShEx correspond to the alternation-free fragment of the modal $\mu$-calculus~\cite{Koz83}. We finally remark that the recent work in \cite{anoukstaticwfs} uses $\mu$-calculus to show decidability and complexity results for static analysis tasks in SHACL under the \emph{well-founded semantics}.


\section{The Multiple Semantics of Recursion}
\label{sec:semantics}


We now introduce the \emph{simple shape language}, which is a core of both ShEx and SHACL, and formally define the three main modes of recursion that are being considered.



\subsection{The RDF Data Model}
\label{sec:thr;sub:rdfdatmod}


An RDF graph is a finite set of triples in $(\IRIs \cup \Blanks) \times (\IRIs) \times (\IRIs \cup \Blanks \cup \Literals)$, where $\IRIs$, $\Literals$ and $\Blanks$ are three countable sets of IRIs, literal data values, and blank nodes, respectively.
If $(u, p, v)$ is a triple of an RDF graph, then $u$ is called its subject, $p$ its predicate, and $v$ its object.
The set of nodes of an RDF graph $\graph$, denoted $\nodes(\graph)$, is the set of IRIs, literals or blank nodes that appear in a subject or an object position in some triple of $\graph$.
We sometimes use \emph{predicates} to refer to the elements of $\IRIs$ when they are used in a predicate position of some triple.

To illustrate, consider the RDF graph $\graph$:  
\[
\graph =
\left\{
\begin{aligned}
  &(\texttt{ex{:}Alice}, \texttt{ex{:}knows}, \texttt{ex{:}Bob}),\\
  &(\texttt{ex{:}Bob}, \texttt{ex{:}age}, 42)
\end{aligned}
\right\}.
\]
Here \texttt{ex{:}Alice} and \texttt{ex{:}Bob} are IRIs used as nodes, 
\texttt{ex{:}knows} and \texttt{ex{:}age} are IRIs used as predicates, 
and $42$ is a literal value.  
In this case $\nodes(\graph) = \{\texttt{ex{:}Alice}, \texttt{ex{:}Bob}, 42\}$.

\subsection{A Simple Shape Language}
\label{sec:thr;sub:simshplan}

Consider a countable set $\Names$ of \emph{shape names}.
A \emph{shape} $\varphi$ is defined by the following syntax, where $s \in
\Names$ is a shape name, $p \in \IRIs$ is a predicate, and $c \in \IRIs \cup
\Literals$ is a constant.
\begin{align*}
    \varphi
  \gDef \
  & \bot
  \gMid \top
  \gMid \hasvalue(c)
  \gMid s
  \gMid \neg \varphi
  \gMid \\ & \varphi \lor \varphi
  \gMid \varphi \land \varphi
  \gMid \exists p. \varphi
  \gMid \forall p. \varphi
  \gEnd
  \end{align*}
A \emph{simple shape language catalogue} or \emph{\textsf{SSL} catalogue} $\cat$ is a partial function mapping shape names to shapes such that $\dom{\cat}$ is finite. We say that the shape name $s$ is \emph{declared} in $\cat$ if $\cat(s)$ is defined; that is,
$s \in \dom{\cat}$. 
We require that every shape name that appears in a shape in (the range of) $\cat$ is also declared in $\cat$. 
For convenience, we usually view $\cat$ as the set of \emph{shape declarations} of the form $\decl{s}{\varphi}$ such that $\cat(s) = \varphi$. By $\const(\cat)$ we denote the set of constants $c$ such that $\hasvalue(c)$ appears in some shape in $\cat$. 

The semantics of a  catalogue $\cat$ on a graph $\graph$ is defined in terms of \emph{shape assignments}, which indicate which shape names hold in which nodes of $\graph$. 
%
A \emph{shape assignment} for $\cat$ and $\graph$ is a relation $\alphaRel \subseteq \dom{\cat} \times (\nodes(\graph) \cup \const(\cat))$.
We define the semantics of shape $\varphi$ in catalogue $\cat$ on graph $\graph$ under assignment $\alpha$ inductively on the structure of $\varphi$ as the set $\semdelta{\varphi}{\alphaRel}{\graph} \subseteq \nodes(\graph) \cup \const(\cat)$ in Table~\ref{tab:simple-shapes-satisfaction}.
We write $\alpha(s)$ for the set $\big\{u \mid (s,u) \in \alpha\big\}$.
We call $\alphaRel$ \emph{correct} if $\alpha(s) = \semdelta{\cat(s)}{\alphaRel}{\graph}$ for every $s \in \dom{\cat}$.


\begin{table}[tb]
\setlength{\tabcolsep}{1pt}
    \centering
    \begin{tabular}{p{2cm}l}
    \toprule
    $\varphi'$ & $\semdelta{\varphi'}{\alphaRel}{\graph}$ \\
    \midrule
    $\bot$ & $\emptyset$\\
    $\top$ & $\nodes(\graph) \cup \const(\cat)$\\
    $\test(c)$ & $\{c\}$\\    
    $s$ & {$\alpha(s)$} \\ 
    $\neg \varphi$ & $\left(\nodes(\graph) \cup \const(\cat)\right) \setminus \semdelta{\varphi}{\alphaRel}{\graph}$\\
    $\varphi_1 \lor \varphi_2$ & $\semdelta{\varphi_1}{\alphaRel}{\graph} \cup \semdelta{\varphi_2}{\alphaRel}{\graph}$\\
    $\varphi_1 \land \varphi_2$ & $\semdelta{\varphi_1}{\alphaRel}{\graph} \cap \semdelta{\varphi_2}{\alphaRel}{\graph}$\\
    $\exists p.\varphi$ & 
    $\big\{u  \;\big\vert\;   \text{for some } (u,p,v) \in \graph, v \in \semdelta{\varphi}{\alphaRel}{\graph} \big\}$ \\
    $\forall p.\varphi$ & $\big\{u \;\big\vert\; \text{for all } (u, p,v) \in \graph, v \in \semdelta{\varphi}{\alphaRel}{\graph} \big\}$ \\
    \bottomrule
    \end{tabular}
    \caption{Semantics of shapes.}
    \label{tab:simple-shapes-satisfaction}
\end{table}

\begin{example}
\label{ex:assignments}
Consider the graph $\graph$ and shape catalogue $\cat$ in Fig.~\ref{fig:example-shape-assignment}.
Intuitively, the shape $\cat(s_1)$ captures the literal "d", $\cat(s_2)$ nodes having a $q$-edge to a node of shape $s_1$, and $\cat(s_3)$ nodes that have a $p$-edge to a node of shape $s_2$.
The assignment $\alphaRel_1 = \{(s_1,\text{"d"}),(s_2,b),(s_3,a)\}$ is correct.
The assignment $\alphaRel_2 = \{(s_1,\text{"d"}), \allowbreak (s_2,b),(s_3,a),(s_3,b)\}$ is not correct because of the pair $(s_3, b)$.
Indeed, under assignment $\alphaRel_2$, only the node $a$ has a $p$-edge to a node assigned to $s_2$.
That is, $\alpha_2(s_3) = \{a,b\} \neq \{a\} = \semdelta{\cat(s_3)}{\alphaRel_2}\graph$.
\end{example}

\begin{figure}
    \centering
    \begin{minipage}{0.3\linewidth}
        \begin{tikzpicture}[scale=0.7,->, >=Stealth, node distance=2cm, every node/.style = {font=\normalsize}]

        \node (a) at (0,0) {$a$};
        \node (b) at (2,0) {$b$};
        \node (d) at (4,0) {"d"};
    
        \path (a) edge [bend left] node[above] {$p$} (b);
        \path (b) edge [bend left] node[below] {$p$} (a);
        \path (b) edge node[above] {$q$} (d);
        \end{tikzpicture}
    %
    \end{minipage}
    \hspace{0.08\linewidth}
    \begin{minipage}{0.55\linewidth}
    \begin{align*}
    \cat = \{ & \decl{s_1}{\test(\text{"d"})},\\
    & \decl{s_2}{\exists q.s_1}, \ \  \decl{s_3}{\exists p.s_2}\}
    \end{align*}
    \end{minipage}
    \caption{A graph $\graph$ and a catalogue $\cat$ with $a,b,p,q \in \IRIs$, $\text{"}\mathrm{d}\text{"} \in \Literals$, and $s_1,s_2,s_3 \in \Names$.}
    \label{fig:example-shape-assignment}
    \vspace{-4mm}
\end{figure}
Next, we consider an example of a recursive catalogue.
\begin{example}\label{ex:gentlerecursion}
Consider 
graph $\graph$ in Fig.~\ref{fig:example-shape-assignment} and 
catalogue $\cat = \{\decl{s}{\exists p.s}\}$. Then
assignment $\alpha_3 = \{(s,a),(s,b)\}$ is correct.
\end{example}
Indeed, the catalogue $\cat$ from Example~\ref{ex:gentlerecursion} is recursive, because the shape in the declaration of $s$ uses $s$ again. In general, a catalogue $\cat$ is recursive if there is a directed cycle in the graph with nodes $\dom{\cat}$ and having an edge $(s,t)$ if and only if $t$ occurs in $\cat(s)$.

\SSL is intentionally simple, yet sufficient to illustrate how the choice of the semantics for recursion influences the expressive power of shape catalogues. As we will see later, the RDF schema languages ShEx and SHACL have more features, in particular for constraining the immediate neighbourhood of nodes (which in \SSL is reduced to $\exists p. \varphi$ and $\forall p. \varphi$ 
for testing the existence of a predicate whose object satisfies a property). 
We discuss in Section~\ref{sec:part2} how our definition relates to real-world ShEx and SHACL.




\subsection{Classic Recursive Semantics}
\label{ssec:semantics}\label{sec:classic-recursive}


We now describe three classic semantics for recursion: \emph{supported model semantics} (\SMS), \emph{greatest fixpoint} (\GFP), and \emph{least fixpoint} (\LFP). 
%
Each semantics specifies which shape assignments \emph{conform} to the shape catalogue. The most permissive is \SMS: it allows all correct shape assignments. 

\begin{definition}
A shape assignment $\alphaRel$ for catalogue~$\cat$ and graph~$\graph$ \emph{conforms to $\cat$ under the supported model semantics} (\emph{\SMS-conforms}), if $\alphaRel$ is correct wrt.~$\cat$. We write $\denot{\CName}[\GName][\SMS]$ for the set of all such $\alpha$.
\end{definition}

The downside of the supported model semantics is that 
it allows 
multiple conforming shape assignments, which can be counter-intuitive and computationally expensive. 

\begin{example}
\label{ex:multiple-correct}
Continuing Ex.~\ref{ex:gentlerecursion} with catalogue $\cat = \{\decl{s}{\exists p.s}\}$, notice that in addition to $\alpha_3$, also $\alpha_4 = \emptyset$ is correct. 
\end{example}


Arbitrarily choosing a single canonical shape assignment among the correct ones is problematic. Two natural ideas that come to mind are: be minimalistic and pick the \emph{smallest} correct assignment (with respect to set inclusion), or be generous and pick the \emph{largest} correct assignment. These ideas ultimately lead to two classical semantics,  \LFP and \GFP, but getting there requires some work because in general there might be multiple smallest and multiple largest correct shape assignments, which we illustrate next. 

\begin{example}
Consider $\graph$ in Fig.~\ref{fig:example-shape-assignment} with catalogue $\cat = \{\decl{s}{\exists p.\, \neg{s}}\}$. Both assignments $\alpha_1 = \{(s,a)\}$ and $\alpha_2 = \{(s,b)\}$ are correct, but neither $\beta_1 = \emptyset$ nor $\beta_2 = \{(s,a),(s,b)\}$ are, showing that $\alpha_1$ and $\alpha_2$ are simultaneously smallest and largest at the same time.
%
\end{example}

Single smallest and largest correct assignments are guaranteed in a special case: when the catalogue does not use negation.\footnote{{There are also weaker  sufficient conditions for monotonicity.}} This follows from the monotonicity of an associated operator ``$\mapsto$'', mapping shape assignments to shape assignments, defined as 
\[\alphaRel \mapsto \big\{(s,v) \in \dom{\cat} \times \mathsf{NC} \mid v \in \semdelta{\cat(s)}{\alphaRel}{\graph}\big\}\]
\[ \text{where } \mathsf{NC}  = \nodes(\graph) \cup \const(\cat)\,.\]
Indeed, it is not hard to see that a shape assignment $\alphaRel$ is correct iff it is a fixpoint of this operator, that is, if $\alphaRel \mapsto \alphaRel$. When the catalogue does not use negation, the operator is monotone and, by the \emph{Knaster–Tarski theorem}~\cite{Tar55}, it has unique least and greatest fixpoints.  Moreover, by the \emph{Kleene fixpoint theorem}~\cite{Tar55}, the least and greatest fixpoint can be computed by iterating the operator on, respectively, the empty shape assignment, 
\[ \emptyset \mapsto \alpha_1 \mapsto \alpha_2 \mapsto \dots\,,\]
and  the ``full'' shape assignment $\Omega = \dom{\cat} \times (\nodes(\graph) \cup \const(\cat))$, 
\[ \Omega  \mapsto \alpha_1 \mapsto \alpha_2 \mapsto \dots\,.\]
This allows us to define the least and greatest fixpoint semantics for catalogues that do not use negation. 

\begin{definition}
Let $\cat$ be a catalogue that does not use negation. 
A shape assignment $\alphaRel$ for catalogue $\cat$ and graph $\graph$ \emph{conforms to $\cat$ under the \LFP (resp.\ \GFP) semantics}, if $\alphaRel$ is the smallest (resp.\ largest) correct assignment w.r.t. $\cat$. We 
also say  that $\alphaRel$ \LFP-conforms (resp. \GFP-conforms) to $\cat$. 
\end{definition}
\noindent Examples~\ref{ex:gentlerecursion} and \ref{ex:multiple-correct} show that the \LFP and \GFP semantics of the same catalogue $\cat$ can be different. 

There is a canonical way of extending any semantics from catalogues without negation to catalogues where  negation is used in a controlled fashion. This approach is called \emph{stratification} and we explain it next. 
Consider a catalogue $\cat$ without negation and assume that we have a unique shape assignment that conforms to $\cat$ (as is the case for the \LFP and \GFP semantics). That is, we know what each shape name in $\cat$ means.  We could now allow using these shape names in declarations of new shape names, both positively and negatively, and when defining the semantics of the new shape names, treat the semantics of the old ones as fixed (called \emph{frozen} in the literature). Let us illustrate this with an example. 

\begin{example}\label{ex:frozen}
Under $\LFP$, the shape assignment that conforms to the catalogue $\{\decl{r}{\test(a) \lor \exists p. r}\}$ associates shape $r$ to nodes that can reach node $a$ using $p$-edges; for example, on graph $\graph$ of Fig.~\ref{fig:example-shape-assignment}, we get $\{(r,a),(r,b)\}$. Let us fix this meaning of $r$ and extend the catalogue  with declaration $s:  \lnot r \lor \exists q.s$. If we now apply the $\LFP$ semantics, with the meaning of $r$ frozen, as described above, the conforming assignment associates $s$ to nodes that can reach via $q$-edges a node that is \emph{safe}, in the sense that it cannot reach $a$ via $p$-edges. On graph $\graph$ in Fig.~\ref{fig:example-shape-assignment} we get  $\{(r, a), (r, b), (s,b), (s,\text{"d"}) \}$. 
\end{example}



Following this idea one can generalize the \LFP and \GFP semantics to \emph{stratified} catalogues, defined below. 

\begin{definition}\label{def:stratification}
A \emph{stratification} for a catalogue $\cat$ is a partitioning of $\dom{\cat}$ into sets $\Sigma_0, \Sigma_1, \dots, \Sigma_n$, called \emph{strata}, such that for all $i$, declarations of shape names from stratum $\Sigma_i$ only use shape names from $\Sigma_0 \cup\Sigma_1 \cup \dots \cup \Sigma_i$, and are allowed to use negation only in expressions of the form $\lnot\hasvalue(c)$ with $c \in \IRIs \cup
\Literals$ and $\lnot s$ with $s\in \Sigma_0 \cup \Sigma_1 \cup \dots \cup \Sigma_{i-1}$. A catalogue $\cat$ is \emph{stratified} if there is a stratification for $\cat$. 
\end{definition}

In Example~\ref{ex:frozen}, we can take $\Sigma_0 = \{r\}$ and $\Sigma_1=\{s\}$  for the extended catalogue.  
Given a stratification $\Sigma_0, \Sigma_1, \dots, \Sigma_n$ for a catalogue $\cat$, we can define the conforming shape assignment inductively for shape names from $\Sigma_i$ for $i=0, 1, \dots, n$, at each level substituting shape names from lower levels with their already computed semantics. While stratifications are not unique, 
the resulting $\LFP$ and $\GFP$ semantics do not depend on their choice (see, e.g., \cite{Apt88}).

\begin{definition} 
\label{def:fixpoint-stratified}
Consider a graph $\graph$ and a catalogue $\cat$ with stratification $\Sigma_0, \Sigma_1, \dots, \Sigma_n$. Let $\cat_i$ be obtained by restricting $\cat$ to $\Sigma_0 \cup \Sigma_1 \cup \dots \cup \Sigma_i$; that is, $\cat_i$ consists of declarations from $\cat$ for shape names from $\Sigma_0 \cup \Sigma_1 \cup \dots \cup \Sigma_i$. 
Define \[\alpha_i \subseteq \Sigma_i \times (\nodes(\graph)\cup\const(\cat))\] iteratively for $i=0, 1,\dots, n$, as follows.
Assuming that $\alpha_0, \dots, \alpha_{i-1}$ are already known, 
let $\alpha_i$ be the least shape assignment over $\Sigma_i$ such that $\alpha_0 \cup \alpha_1 \cup\dots \cup \alpha_i$ is correct for $\cat_i$. 
We say an assignment $\alpha$ for $\graph$ and $\cat$ 
\emph{\LFP-conforms to $\cat$} if $\alpha = \alpha_0 \cup \alpha_1 \cup\dots \cup \alpha_n$. Notice that, by definition, $\alpha$ is unique. We write this $\alpha$ as $\denot{\CName}[\GName][\LFP]$.
The shape assignment  $\denot{\CName}[\GName][\GFP]$  
\emph{\GFP-conforming to $\cat$}
is defined analogously, except that for $\alpha_i$ we take the largest shape assignment over $\Sigma_i$ s.t. $\alpha_0 \cup \alpha_1 \cup\dots \cup \alpha_i$ is correct for $\cat_i$. 
\end{definition}

\begin{example}\label{ex:GFP}
Shape assignments that \GFP-conform to catalogue $\{\decl{s}{\lnot \test(b) \land \forall p. s}\}$  associate shape $s$ to nodes that cannot reach node $b$ using $p$-edges. For $\graph$ from Fig.~\ref{fig:example-shape-assignment}, one gets $\{(s,\text{"d"})\}$.
\end{example}

The \LFP and \GFP semantics are tightly connected by a fundamental principle of \emph{duality}, which allows translations back and forth. 
%
For a stratified \SSL catalogue $\cat$ we define the \emph{dual catalogue}  \[\widetilde{\cat} = \big\{\decl{s}{\widetilde{\varphi}} \mid (\decl{s}{\varphi}) \in \cat \big\}\]
where the dual shape $\widetilde\varphi$ is defined recursively as follows, relying on $\varphi$ using $\lnot$ only in expressions 
$\lnot\hasvalue(c)$ and $\lnot s$, 
\begin{gather*}
\widetilde\bot = \top\,, \  \widetilde\top = \bot\,, \
\widetilde{\hasvalue(c)} = \lnot \hasvalue(c)\,, \ \widetilde{\lnot\hasvalue(c)} = \hasvalue(c)\,, \\
\widetilde{s} = s\,, \ \widetilde{\lnot s} = \lnot s\,, \ \widetilde{\varphi_1 \land \varphi_2} = \widetilde{\varphi_1} \lor
  \widetilde{\varphi_2}\,,\\  \widetilde{\varphi_1 \lor \varphi_2} =
  \widetilde{\varphi_1} \land \widetilde{\varphi_2}\,, 
\widetilde{\exists p.\varphi} = \forall p.\widetilde{\varphi}\,, \quad \widetilde{\forall p.  \varphi} =
  \exists p. \widetilde{\varphi}\,.
\end{gather*}
That is, $\dual*{\varphi}$ is obtained from $\varphi$ by swapping $\top$ and $\bot$,  $\hasvalue(c)$ and $\lnot\hasvalue(c)$, $\land$ and $\lor$, as well as $\exists$ and $\forall$, but $s$ and $\lnot s$ are not swapped.


\begin{example} 
Consider again the catalogue $\{\decl{s}{\lnot \test(b) \land \forall p. s}\}$ from Ex.~\ref{ex:GFP}. Shape assignments that \LFP-conform to its dual $\{\decl{s}{\test(b) \lor \exists p. s}\}$ assign $s$ to nodes that can reach $b$ using $p$-edges.
\end{example}

\begin{proposition}
\label{prop:duality-lfp-gfp}
Let $\graph$ be a graph and $\cat$ a stratified \SSL catalogue. Then, for every shape assignment $\alpha$ for $\graph$ and $\cat$,
\begin{center}
$\alpha$ \LFP-conforms to $\cat$ \quad \iff \quad 
$\Omega\setminus\alpha\:$  \GFP-conforms to $\widetilde\cat$
\end{center}
for $\Omega=\dom\cat \times (\nodes(\graph) \cup \const(\cat))$. 
\end{proposition}

\begin{proof}
Suppose first that $\cat$ has a single stratum; that is, it only uses $\lnot$ in expressions of the form $\lnot\hasvalue(c)$. Let us write $\mapsto$ for the operator associated with $\cat$ and $\tilde\mapsto$ for the one associated with $\widetilde\cat$. Both operators are monotone. 
From the definition of $\widetilde\cat$ it follows that 
$\alpha \mapsto \alpha'$ iff $\Omega\setminus \alpha  \mathbin{\tilde\mapsto} \Omega \setminus \alpha'$. Hence,  
$
\emptyset \mapsto \alpha_1 \mapsto \alpha_2 \mapsto \dots$ iff
$\Omega \mapsto \Omega \setminus \alpha_1 \mathbin{\tilde\mapsto} \Omega \setminus \alpha_2 \mathbin{\tilde\mapsto} \dots$. By the Kleene Fixed-Point Theorem, $\alpha$ is  the least fixpoint of $\mapsto$ iff $\Omega\setminus \alpha$ is the greatest fixpoint of $\tilde\mapsto$.  

For a general stratified catalogue $\cat$, we note that a stratification $\Sigma_0, \Sigma_1, \dots, \Sigma_n$ for $\cat$ is also a stratification for $\widetilde \cat$, and proceed by straightforward induction following Def.~\ref{def:fixpoint-stratified}.
\end{proof}

Using Proposition~\ref{prop:duality-lfp-gfp} one can show that $\SSL$ under both $\LFP$ and $\GFP$ corresponds precisely to the alternation-free fragment of the modal $\mu$-calculus \cite{Koz83} with nominals but without atomic propositions, which means it has considerable expressive power.



\subsection{Schemas, Conformance, and Validation}
\label{sec:validation}
The real purpose of shape languages is specifying conformance of \emph{graphs} to  \emph{shape schemas}, which ensures that graphs satisfy some specified constraints. The task of checking if a given graph conforms to a schema is called \emph{validation}.

To this end, both SHACL and ShEx have a mechanism to select nodes that need to satisfy some shape: SHACL has \emph{target declarations}~\cite{KK17} and ShEx has \emph{shape maps}~\cite{PB19}. We abstract them here as \emph{selector maps}. Given a shape catalogue $\cat$, a selector map $\sel$ for $\cat$ is a finite set of \emph{selectors} of the form $\decl{s}{\sigma}$ with $s\in\dom{\cat}$ 
and $\sigma$ is a shape that does not use shape names.\footnote{Both ShEx and SHACL impose additional restrictions on $\sigma$, essentially allowing only atomic shape expressions, and do not allow $\top$; see~\cite{www25}.} A \emph{shape schema} is then a tuple $\schema=(\cat,\sel)$ with $\cat$ a shape catalogue, and $\sel$ a selector map.  Intuitively, every $\decl{s}{\sigma}$ in  $\sel$ expresses that every node $u$ that has shape $\sigma$ should also have shape $s$ under shape assignment(s) conforming to $\cat$. In order to define the semantics precisely, we need to handle the nondeterminism of $\SMS$ and agree on the domain over which $u$ ranges.


A shape assignment for catalogue $\cat$ and graph $\graph$ associates shape names with elements of  
$\nodes(\graph) \cup \const(\cat)$, which does not necessarily include all constants mentioned in $\sel$. We deal with this by incorporating the selector map into the catalogue as follows. 
We define $\cat_\sel = \cat \cup \{\decl{t_{s,\sigma}}{\sigma \land \lnot s} \mid \decl{s}{\sigma} \in \sel\}$, where every $t_{s,\sigma}$ is a fresh shape name that depends on $s$ and $\sigma$. 
Then, for a correct shape assignment $\alpha$ for $\graph$ and $\cat_\sel$,
we say a graph $\graph$ \emph{conforms} to schema $(\cat,\sel)$ under  $\alpha$ if $\alpha(t_{s,\sigma}) = \emptyset$ for every shape name $t_{s,\sigma}$.
Finally, we say that $\graph$ \emph{$\LFP$-conforms} to  $(\cat,\sel)$
if $\graph$ conforms to $(\cat,\sel)$ under $\alpha = \denot{\CName_\sel}[\GName][\LFP]$, and analogously for \GFP. 

We can define graph conformance under \SMS similarly, but we need to decide what to do with the multiple conforming shape assignments a catalogue might admit. Two approaches have been considered, called the \emph{cautious} and \emph{brave} \SMS semantics~\cite{CRS18,ACORSS20,BJ21}. We say that $\GName$ \emph{cautiously} (resp.\ \emph{bravely}) \emph{\SMS-conforms to} schema $(\cat,\sel)$ if 
$\graph$ conforms to $(\cat,\sel)$
under 
\emph{every} $\alpha \in \denot{\CName_\sel}[\GName][\SMS]$ (resp.\ \emph{some} $\alpha \in \denot{\CName_\sel}[\GName][\SMS]$). 
We use brave semantics by default and say simply that \emph{$\graph$ \SMS-conforms to $\cat$}.

\begin{example}\label{ex:separation-1}
Consider the graph $\graph = \{ (a,p,a) \}$ and the schema $(\mathcal{C}, \sel)$ where $\mathcal{C} = \{ s: \exists p. s\}$ and  $\sel = \{ s : \test(a)\}$. Under $\LFP$ semantics, we would expect the shape assignment  $\alpha_l = \emptyset$, while under $\GFP$ semantics, we get the shape assignment $\alpha _g = \{ (s,a) \}$.
Thus, $\graph$ $\GFP$-conforms to $(\mathcal{C},\sel)$, but $\graph$ does not $\LFP$-conform to
$(\mathcal{C},\sel)$.
\end{example}


\section{A Comparative Study of Implementations}
\label{sec:thr;sub:comstdschval}

\NewDocumentCommand{\rot}{O{45} O{1em} m}{\makebox[#2][l]{\rotatebox{#1}{#3}}}%

We present an experimental study of real-world engines for SHACL and ShEx to compare how they treat recursive schemas. The source code and data needed to reproduce the results of the study, as well as the raw results of our experiments, are available on Zenodo \cite{sheval}.

We want to understand, for each engine, \emph{is there a formal semantics for recursive shapes that explains its behaviour?}
The main challenge is that real-world systems only test \emph{whether a given graph conforms to a given schema}, 
without giving access to the witnessing shape assignments.

In order to detect which semantics is applied, we design test cases that aim to separate $\LFP$, $\GFP$, brave $\SMS$, and cautious $\SMS$. First, we run all test cases on all engines.
If we observe inconsistent behaviour, we investigate the test results in more detail. If all answers are consistent with one of the four considered semantics, we interpret this as the answer to our question.

We designed thirteen test cases, falling into two test categories: separation and feature tests. Separation tests attempt to distinguish whether an engine uses \LFP, \GFP, brave \SMS, or cautious \SMS. Feature tests do exactly what their name suggests. Each such test $S  = \langle \graph, \cat, \sel \rangle$ consists of a graph $\graph$, a shape catalogue $\cat$, and a shape map $\sel$. The intended use of $S$ is to produce a ShEx or SHACL schema based on $\cat$ and $\sel$ to check whether $\graph$ conforms to $(C,\sel)$.


For each separation test we list all correct shape assignments. The unique conforming assignments under $\LFP$ and $\GFP$ appear as $\alpha_{\LFP}$ and $\alpha_{\GFP}$, and the remaining ones as $\alpha_1, \alpha_2$. 
An assignment $\alpha$ is {\pass blue}, if $\graph$ conforms to $(\cat, \sel)$ under $\alpha$,  and  {\fail red}, if it does not.

\paragraph{Basic separation tests.}
These four tests are designed to distinguish a validator that employs $\LFP$ semantics from one that employs $\GFP$, using minimal examples. 
%


{ \small
\begin{itemize}
\item \texttt{bsep1}: \\
$\begin{array}[t]{l}
\graph = \{ (a,p,a) \}\,,\ \mathcal{C} =  \{ \decl{s}{\exists p . s} \}\,,\ 
\sel = \{ \decl{s}{\test(a)}\}\,, \\ 
\fail{\alpha_{\LFP} = \emptyset}\,,\ \pass{\alpha_{\GFP} = \{(s,a)\}}\,.
\end{array}$

\item \testcode{bsep2}: \\
$\begin{array}[t]{l} 
 \graph = \{(a,p,a), (b,p,b) \}\text{, } \mathcal{C} =  \{\decl{s}{\exists p . s} \}\text{, }\\
 \sel = \{ \decl{s}{\test(a)}, \decl{s}{\test(b)}\}\,,\ \\
 {\fail \alpha_{\LFP}=\emptyset}\,,\ 
 {\pass \alpha_{\GFP}=\{(s,a),(s,b)\}}\,,\ \\
 {\fail \alpha_1=\{(s,a)\}}\,,\ 
 {\fail \alpha_2 =\{(s,b)\}}\,.   
\end{array}$
\item  \testcode{bsep3}: \\
$
\begin{array}[t]{l}
\graph = \{(a,p,c), (b,p,c) \}\text{, }\mathcal{C} =  \{\decl{s}{s' \land \exists p},\ \decl{s'}{s \land \exists p} \},\\ \sel = \{ \decl{s}{\test(a)}, \decl{s}{\test(b)}\}\,,\\
{\fail \alpha_{\LFP} = \emptyset}\,,\ 
{\pass \alpha_{\GFP}  = \{(s,a), (s,b), (s',a), (s',b)\}}\,,\ \\
{\fail \alpha_1 = \{(s,a), (s',a)\}}\,,\ 
{\fail \alpha_2 = \{(s,b), (s',b)\}}\,.  
\end{array}$
\item \testcode{bsep4}:  \\
$
\begin{array}[t]{l}
\graph = \{ (a,p,a), (b,p,b) \},\\ 
\mathcal{C} = \{\decl{s}{\exists p.s}, \decl{s'}{\neg s} \},\ \sel = \{ \decl{s}{\test(a)}, \decl{s}{\test(b)}\}\,,\\
{\fail \alpha_{\LFP} = \{(s',a), (s',b)\}}\,,\ 
{\pass \alpha_{\GFP} =  \{(s,a), (s,b) \}}\,, \\    
{\fail \alpha_1 = \{(s,a), (s',b) \}}\,,\ 
{\fail \alpha_2 = \{(s,b), (s',a) \}}\,. 
\end{array}$      
\end{itemize}
}
\paragraph{Reachability separation tests.}
Then, we tested a classical recursive property of reachability and its dual property of safety by checking them on two dual shape names $r$ and $s$, in four different scenarios. All scenarios use the same graph:\\
\begin{minipage}{\linewidth}
\begin{tikzpicture}[scale=0.7,->, >=Stealth, node distance=2cm, every node/.style = {font=\normalsize}]

\node (g) at (-1,0) {$\graph$:};

\node (a) at (0,0) {$a$};
\node (b) at (3,-.3) {$b$};
\node (c) at (7,0) {$c$};
\node (d) at (3,.5) {$d$};
    
\path (a) edge [out=20,in=180] 
    (d);
\path (d) edge [out=0,in=160] 
    (c);
\path (b) edge [out=180,in=-13] 
    (a);
\path (b) edge [out=0,in=193] 
    (c);
\path (c) edge [out=177,in=-15] 
    (d);

\path (8,.5) edge (9,.5);
\node at (9.8,.5) {:\ $p$-edge};
\end{tikzpicture}
\end{minipage}
{ \small
\begin{itemize}
\item \testcode{reach1}:  \\
$
\begin{array}[t]{l} 
\mathcal{C} =  \{\decl{r}{\test(a) \lor \exists p.r}  \},\\ 
\sel = \{ \decl{r}{\test(a)}, \decl{r}{\test(b)}, \decl{r}{\test(c)}, \decl{r}{\test(d)}\}\,,\\
{\fail \alpha_{\LFP}= \{(r,a), (r,b)\}} \,,\\
{\pass \alpha_{\GFP}=\{(r,a),(r,b),(r,c),(r,d)\}}\,. \\ 
\end{array}$

\item \testcode{reach2}: \\
$\begin{array}[t]{l}
\mathcal{C} =  \{\decl{s}{\neg \test(a) \land \forall p. s},\ \decl{r}{\neg s} \}, \\
\sel = \{ \decl{r}{\test(c)}, \decl{r}{\test(d)}\}\,,\\
{\pass \alpha_{\LFP}=\{(r,a),(r,b),(r,c),(r,d)\}}\,, \\
{\fail \alpha_{\GFP}=\{(r,a), (r,b),(s,c),(s,d)\}}\,. 
\end{array}$

\item \testcode{safe1}: \\ 
$\begin{array}[t]{l}
\mathcal{C} =  \{\decl{s}{\neg \test(a) \land \forall p. s} \}\text{, }
\sel = \{ \decl{s}{\test(c)}, \decl{s}{\test(d)}\}\,,\\
{\fail \alpha_{\LFP}=\emptyset}\,,\ 
{\pass \alpha_{\GFP}=\{(s,c)(s,d) \}}\,.     
\end{array}$

\item \testcode{safe2}: \\ 
$
\begin{array}[t]{l}
  \mathcal{C} =  \{\decl{r}{\test(a) \lor \exists p. r},\ \decl{s}{\neg r} \},\\
  \sel = \{ \decl{r}{\test(c)}, \decl{r}{\test(d)}\}\,,\\
{\fail \alpha_{\LFP} = \{(r,a), (r,b),(s,c),(s,d)\}}\,,\ \\
{\pass \alpha_{\GFP} = \{(r,a),(r,b),(r,c),(r,d)\}}\,.
\end{array}$
\end{itemize}

}
\paragraph{Feature tests} These target more specific properties of the validator, as we will explain after introducing them.
{ \small
\begin{itemize}
\item \texttt{nstrat1} (non-stratified negation without cyclic data): \\
$ \graph = \{ (a,p,b), (b,p,c) ,(c,p,d), (d,p,e)  \}, 
  \mathcal{C}=\{ \decl{s}{\exists p.\neg s}  \}$,  \\
  $\sel = \{ \decl{s}{\test(a)},\ \decl{s}{\test(b)},\ \decl{s}{\test(c)},\ \decl{s}{\test(d)},\ \decl{s}{\test(e)} \}
$\\
\emph{Passing Condition}: Accept or reject the schema over $\graph_{n1}$. 

\item \testcode{nstrat2} (non-stratified negation with cyclic data): \\
$\graph = \{ (a,p,b), (b,p,a)  \}, \mathcal{C} =\{ \decl{s}{\exists p.\neg s} \}$, \\
   $\sel = \{ \decl{s}{\test(a)},\decl{s}{\test(b)}\}$\\
\emph{Passing Condition}:  Accept or reject the schema over $\graph$.

\item \texttt{fresh} (fresh constant support):\\ $\graph  = \{  (a,p,b), \allowbreak (b,p,c), \allowbreak (c,p,a) \}$, $\mathcal{C} = \{ \decl{s}{\top} \}$, \\ 
$\sel = \{ \decl{s}{\test(d)} \}$\\
\emph{Passing Condition}: The validator accepts. 

\item \texttt{cons1} (consistency): $\graph = \{ (a,p,a)\}$, $\mathcal{C}=\{ \decl{s}{\exists p.s'}, \decl{s'}{\neg s} \}$,\\ $\sel = \{ \decl{s}{\test(a)}, \decl{s'}{\test(a)}\}$\\
\emph{Passing Condition}: The validator rejects.

\item \testcode{cons2} (consistency): $\graph \text{ as in \testcode{nstrat2}},\ \mathcal{C} \text{ as in \testcode{cons1}}$,\\ 
$\sel = \{ \decl{s}{\test(a)},\decl{s'}{\test(a)},\decl{s}{\test(b)},\decl{s'}{\test(b)} \}$\\
\emph{Passing Condition}: The validator rejects.
\end{itemize}
}
The tests \testcode{nstrat1} and \testcode{nstrat2} are for shape catalogues that do not permit a stratification as in \Cref{def:stratification}. While none of our semantics support this setting, we still aim to find out how real-world validators behave on such inputs. The test \testcode{fresh} has a shape catalogue with a shape assignment that is trivially satisfied, using the shape $\top$. In the shape map, we require that this shape is satisfied in a ``fresh node'', that is, a node that is not featured in the graph. This reflects something that both SHACL and ShEx permit: selecting nodes outside the input graph. 
Finally, \testcode{cons1} and \testcode{cons2} check for \emph{logical consistency}. We have a catalogue with two shapes, $s$ and $s'$, where $s'$ is defined as the negation of $s$. Then we ask in the shape map that both $s$ and $s'$ be assigned to the same node. Since no such assignment is possible while being consistent with the catalogue's semantics, this test requires that a validator rejects.



\paragraph{Validation Engines}
Since the SHACL W3C Recommendation does not require SHACL validators to support recursive schemas, many validators reject them all or raise a warning indicating that recursion is not supported. We focus on those that do support recursive SHACL schemas or where we can ignore the warnings and perform the validation anyway, meaning that they accept a non-empty set of graphs when validating against a recursive SHACL schema. 
These are pySHACL~\cite{pySHACL},  SHACL-S~\cite{shacls}, Jena SHACL~\cite{jenacl},  Topbraid~\cite{tq}, and rudof~\cite{rudof}, which supports both SHACL and ShEx. 
For ShEx, we have ShEx-S~\cite{shexs}, Jena ShEx~\cite{jenacl}, and rudof~\cite{rudof}.

\begin{table}
\small
\setlength{\tabcolsep}{1pt}
\centering 
\begin{tabular}{l@{\hspace{2ex}}ccc@{\hspace{3mm}}ccccc@{\hspace{4mm}}cccc}
    \toprule
    \rot{Test instance}  & \rot{rudof (ShEx)} & \rot{Jena ShEx} & \rot{ShEx-S}  & \rot{pySHACL} & \rot{SHACL-S} & \rot{Jena SHACL} & \rot{Topbraid}  & \rot{rudof (SHACL)} & \rot{$\GFP$} &   \rot{$\LFP$} & \rot{b$\SMS$} & \rot{c$\SMS$} \\
     \midrule
     \texttt{bsep1}    &     \y &        \y &         \y & \y      &      \y &         \y &   \y &           \y      &\y & \n & \y & \n \\ 
     \texttt{bsep2}    &     \y &        \y &         \y & \y      &      \y &         \y &   \y &           \y       &\y & \n & \y & \n\\ 
     \texttt{bsep3}    &     \y &        \y &            \y & X      &      X &         \y &   \n &           X      &\y& \n &\y&\n\\ 
    \texttt{bsep4}     &     \y &        \y &            \y & \y      &      \y &         \y &   \y &           \y     &\y& \n &\y&\n\\ 
     \texttt{reach1}  &     \y &        \y &            \y & \y      &      \y &         \y &   \n &           \y     &\y& \n &\y&\n\\ 
     \texttt{reach2} &     \n &        \n &            \n & \y      &      \y &         \y &   \y &           X     &\n& \y &\y&\n\\ 
     \texttt{safe1} &     \y &        \y &            \y & \y      &      \y &         \y &   \y &           X     & \y & \n & \y & \n \\ 
     \texttt{safe2}   &     \y &        \y &            \y & \y      &      \y &         \y &   \n &           \y     & \y & \n & \y & \n\\ 
     \midrule

     \texttt{nstrat1}    & \xxmark & \cmark &  \xxmark & \cmark  & \cmark & \cmark & \cmark & \cmark & -  & - &  - & - \\ 
     \texttt{nstrat2}    & \xxmark & \cmark &  \xxmark & \cmark  & \cmark & \cmark & \cmark & \cmark  & -  & - &  - & -\\ 
     \texttt{fresh}       & \xmark  & \cmark &  \cmark     & \cmark  & \cmark & \cmark & \cmark & \cmark  & -  & - &  - & - \\      
     \texttt{cons1}       & \xxmark & \cmark &  \xxmark     & \cmark  & \cmark & \cmark & \cmark & \cmark  & -  & - &  - & -\\ 
     \texttt{cons2}       & \xxmark & \xmark &  \xxmark     & \cmark  & \cmark & \cmark & \cmark & \cmark & -  & - &  - & -\\ 
    \bottomrule
\end{tabular}
    \caption{Results of comparative study of the semantics of SHACL and ShEx validators for recursive shape catalogues. }
    \label{tab:exampleExpTable}
\vspace{-4mm}
\end{table}

\paragraph{Discussion} We see the results of our experiments in \Cref{tab:exampleExpTable}. 
For each of the separation tests (first eight rows), a validator can answer  \emph{yes} ({\y}) or \emph{no} ({\n}), or fail (X) by reporting an error, 
crashing, or not terminating within 5 seconds.\footnote{Our test graphs have at most 5 nodes and edges, so the timeout is rather generous.}
In the columns \GFP, \LFP, b\SMS (brave \SMS), and c\SMS (cautious \SMS), we indicate the answer prescribed by the respective semantics. For \texttt{bsep1}, for example, $\graph$ does not conform to $(\cat,\sel)$ under ${\fail \alpha_{\LFP}}$, so the correct answer under $\LFP$ is no (\n). However, $\graph$ does conform to $(\cat,\sel)$ under ${\pass \alpha_{\GFP}}$, so the correct answer under $\GFP$ is yes (\y).  The answer yes is also correct under b\SMS, because ${\pass \alpha_{\GFP}}$ witnesses it. However, the answer no is correct under c\SMS because ${\fail \alpha_{\LFP}}$ witnesses it: it is a correct assignment that does not associate $s$ to $a$ as required by $\sel$. An analogous reasoning applies for all values of yes (\y) and no (\n) in the last four columns.

Therefore, using Table~\ref{tab:exampleExpTable}, we can immediately see which validators are consistent with which semantics. We see that the three ShEx validators, rudof, Jena ShEx and ShEx-S, are consistent with \GFP, as prescribed by the official semantics \cite{shex_standard}. On the SHACL side, the picture is more varied. Four validators, namely pySHACL, SHACL-S, Jena SHACL and rudof, are consistent with b\SMS, if we ignore the error cases. The behaviour of Topbraid stands out. It is not consistent with either \GFP, \LFP, b\SMS, or c\SMS, as can be seen. Topbraid thus seems to follow a unique semantics.


For the five feature tests we have two outcomes, pass ({\pass\ding{51}}) or fail ({\fail\ding{55}}), since by design these tests interpret errors unrelated to validation
as failing the test. We will nonetheless indicate failures due to these causes by an asterisk ({\fail \ding{55}*}). 
The failures for the ShEx validators can be explained by the fact that all feature tests, except for \texttt{fresh}, use non-stratified negation, which is disallowed in ShEx. ShEx-S and rudof (ShEx) correctly reject such inputs. JenaShEx does not reject them and attempts validation. We see that it passes all tests, except for the second test for logical consistency.
Perhaps this asymmetric behaviour of JenaShEx for two very similar tests can be attributed to a software bug. One notable result is that rudof does not pass \texttt{fresh}; indeed it is the only validation engine to fail it, and curiously, rudof itself, when targeting SHACL, \emph{does} pass it. 
On the SHACL side, all validators pass all the feature tests.

Note that our test results do not prove that any of the validators follows a given semantics, as other tests might reveal inconsistencies. They do prove, however, that some validators do not follow a given semantics, and they give us hints as to how these real-world systems act. 
Claiming that their semantics follows \GFP, \LFP, b\SMS, or c\SMS requires knowledge of the algorithm each system implements, and a formal study of said algorithm.


\section{Recursion in ShEx and SHACL}
\label{sec:part2}\label{sec:expressive}

Building on the framework developed in Section~\ref{sec:semantics}, we now explore how the choice of semantics for recursion impacts the expressive power and complexity of ShEx and SHACL. We give semantics-preserving translations between  core fragments of the two languages under $\GFP$/$\LFP$
and establish tight complexity bounds  exposing the cost of $\SMS$.


\subsection{The ShEx and SHACL Schema Languages}
\label{sec:full-shex-shacl}

We instantiate our abstract framework with the standard RDF validation languages ShEx \cite{shex_standard} and SHACL \cite{KK17}.
We see them as variants of \SSL that allow different kinds of shapes: {most notably, SHACL shapes use complex path expressions to specify paths in terms of sequences of predicates, and ShEx shapes use triple expressions to specify neighbourhoods in terms of sets of triples.}
The notions of shape catalogues and shape assignments naturally extend to ShEx and SHACL.
Below we present the syntaxes of these languages as they were presented in \cite{www25}, to which we add recursion via references to shape names, as in \SSL.
Their semantics are combinations of the semantics of (non-recursive) shapes from \cite{www25} with the semantics of recursion we have defined for \SSL. 

Both languages use \emph{value type} constraints that are interpreted as sets of constants. 
Formally, $\ValueTypes$ is a finite set of value types, and for every $\vtype \in \ValueTypes$, its semantics is $\sem{\vtype} \subseteq \IRIs \cup \Literals \cup \Blanks$.
Throughout this subsection, we fix an RDF graph~$\graph$.

\begin{definition}[SHACL syntax and semantics]
\label{def:shacl-syntax-semantics}
\label{def:shacl-syntax}
SHACL \emph{shapes} $\varphi$ and \emph{path expressions} $\pi$ are given by the following grammar, with $c \in \IRIs \cup \Literals$, $\vtype \in \ValueTypes$, $Q \subseteq_{\text{fin}} \IRIs$, 
$p \in \IRIs$, $n \in \mathbb{N}$, $s \in \Names$:
    \begin{align*}
    \varphi
  \gDef \
  & \top \gMid \bot
  \gMid \hasvalue(c)
  \gMid \test(\vtype)
  \gMid s
  \gMid \closed(Q)
  \gMid \neg \varphi
  \gMid \varphi \land \varphi \\
  & \varphi \lor \varphi \gMid 
   \eq(\pathExpr, p)
  \gMid \disj(\pathExpr, p)
  \gMid \geqn{n}{\pathExpr}{\varphi}
  \gMid \leqn{n}{\pathExpr}{\varphi}
  \gEnd \\
  \pi \gDef & \id
  \gMid p
  \gMid \pathExpr^{-}
  \gMid \pathExpr \cdot \pathExpr
  \gMid \pathExpr \cup \pathExpr
  \gMid \pathExpr^{*}
 \gEnd
  \end{align*}
The semantics $\semdelta{\varphi}{\alpha}{\graph} \subseteq \nodes(\graph) \cup \const(\cat)$ of a SHACL shape $\varphi$ in catalogue $\cat$ on graph $\graph$ under assignment $\alpha$ is defined by recursion on the structure of $\varphi$, by extending Table~\ref{tab:simple-shapes-satisfaction} with
\begin{align*}
    \semdelta{\test(\vtype)}{\alpha}{\graph} = &\ 
    \sem{\vtype},\\
    \semdelta{\geqn{n}{\pi}{\varphi} }{\alpha}{\graph}= & \ \big\{v \;\,\big\vert\;\, \#\big\{ u \in \sem{\pi}^{\graph}_v \mid u \in \semdelta{\varphi}{\alpha}{\graph}\big\} \ge n \big\},\\
    \semdelta{\leqn{n}{\pi}{\varphi} }{\alpha}{\graph}= & \ \big\{v \;\,\big\vert\;\, \#\big\{ u \in \sem{\pi}^{\graph}_v \mid u \in \semdelta{\varphi}{\alpha}{\graph}\big\} \le n \big\},\\
    \semdelta{\closed(Q)}{\alpha}{\graph} = &\ \big\{ v \;\big\vert\; 
     \sem{p}^\graph_v = \emptyset\;\text{for all}\; p\in \IRIs\setminus Q\big\}, \\
    \semdelta{\eq(\pi, p)}{\alpha}{\graph} = &\ \big\{v \;\big\vert\; \sem{\pi}^\graph_v = \sem{p}^\graph_v \big\},\\ 
    \semdelta{\disj(\pi, p)}{\alpha}{\graph} = &\ \big\{ v \;\big\vert\; \sem{\pi}^\graph_v \cap \sem{p}^\graph_v  = \emptyset \big\},
\end{align*}
where 
$\sem{\pi}^{\graph}_v$ is the set of nodes accessible in $\graph$ from $v$ by a sequence of edges whose predicates match the path expression $\pi$; that is, $\sem{\pi}^{\graph}_v = \big\{u \;\big\vert\; (v, u) \in \sem{\pi}^{\graph}\big\}$ and $\sem{\pi}^{\graph}$ is as in \cite{www25}.
\end{definition} 


\begin{definition}[ShEx syntax and semantics]
        \label{def:shex-syntax-semantics}
    \label{def:shex-syntax}
    ShEx \emph{shapes} $\varphi$, \emph{triple expressions} $e$, and \emph{closed triple expressions} $f$, are defined by the following grammar, with $c \in \IRIs \cup \Literals$,  $\vtype \in \ValueTypes$, $s \in \Names$, $p \in \IRIs$, and $R, Q \subseteq_{\text{fin}} \IRIs$:
\begin{align*}
		\varphi  \gDef & \test(c)
        \gMid \test(\text{$\vtype$})
        \gMid \shexref s
        \gMid \shexneigh{e}
		\gMid \varphi \land \varphi
		\gMid \varphi \lor \varphi
		\gMid \lnot \varphi
		\gEnd \\
		e \gDef & f \shexeach (\shexneginv{R})^{*} \ \gMid\   f \shexeach (\shexneginv{R})^{*} \shexeach (\shexneg{Q})^{*} \gEnd \\
        f \gDef & \varepsilon
		\gMid p.\varphi
		\gMid \shexinverse p.\varphi
		\gMid f \shexeach f
		\gMid f \shexone f
        \gMid f^{*}
		\gEnd
	\end{align*}
    The semantics  $\semdelta{\varphi}{\alpha}{\graph} \subseteq \nodes(\graph) \cup \const(\cat)$ of a ShEx shape $\varphi$ in catalogue $\cat$ on graph $\graph$ under  assignment $\alpha$ is  defined by recursion on the structure of $\varphi$, by extending Table~\ref{tab:simple-shapes-satisfaction} with
    \[
    \semdelta{\shexneigh{e}}{\alpha}{\graph} = \big\{ \, v \;\big\vert\; \neigh^{\pm}_\graph(v) \in \semdelta{e}{\alpha}{\graph}\big\},
    \]
    where $\neigh^{\pm}_\graph(v)$ 
    collects edges $(v, p, v') \in \graph$ outgoing from $v$ and inverses $(v, p^{-}, v')$ of edges $(v', p, v) \in \graph$ incoming to $v$, 
    and $\semdelta{e}{\alpha}{\graph}$ is the family of sets of edges and inverses of edges from $\graph$ that match $e$, defined recursively as follows. 
    Expression $\varepsilon$ matches $\emptyset$. Expressions $p.\varphi$ and $p^-.\varphi$ match singletons $\{(u,p,w)\}$ and $\{(u,p^-,w)\}$ such that $w \in \semdelta{\varphi}{\alpha}{\graph}$. Expression $f_1 \shexone f_2$ matches sets matched by $f_1$ and sets matched by $f_2$. Expression $f_1 \shexeach f_2$ matches {sets that can be expressed as a union of two} disjoint sets $F_1$ and $F_2$ matched by $f_1$ and $f_2$, respectively. Expression $f^*$ matches {sets that can be expressed as a union of} arbitrarily many disjoint sets matching $f$. Expression $(\lnot Q)^*$ matches arbitrary sets of $p$-edges with $p\notin Q$ and $(\lnot R^-)^*$ matches arbitrary sets of $p^-$-edges with $p\notin R$.
    \end{definition}

{Intuitively, a set of triples matched by a closed triple expression $f$ may contain only triples explicitly mentioned in $f$. A closed expression cannot be used directly in a ShEx shape: it can only occur as a subexpression of an expression $e$ which opens allowed set of  triples to all $p^-$-edges with $p\notin R$ and possibly also to all $p$-edges with $p\notin Q$.}


Stratification for ShEx and SHACL is defined as for \SSL. (Using auxiliary shape definitions, one can rewrite each schema so that it applies negation only to sub-expressions of the form $\test(c)$, $\test(\vtype)$ and $s$.) Following standard ShEx, we assume \GFP by default and say that a graph conforms to a ShEx schema if it \GFP-conforms.
\subsection{Expressive Power}
\label{sec:comparison-shex-shacl}

We now identify large equi-expressive syntactic fragments of ShEx (with \GFP) and SHACL with \LFP that have the same expressive power.
In the case of ShEx, we forbid applying `*' to triple expressions that contain  `$\shexeach$'. 

\begin{definition}[Restricted ShEx] \emph{Restricted ShEx} is the syntactic fragment of ShEx obtained by replacing the rule 
for closed triple expressions $f$ from Def.~\ref{def:shex-syntax} with the following derivation rules:
\begin{align*}
    f \gDef & f' \gMid f'^{*} \gMid f \shexeach f \gMid f \shexone f
    \gEnd \\
    f' \gDef &
    \varepsilon \gMid p.\varphi \gMid \shexinverse{p}.\varphi \gMid f' \shexone f'
    \gEnd 
\end{align*}
\end{definition}

In the case of SHACL, we restrict counting to single edges, but allow existential quantification over arbitrary path expressions. We also disallow  $\eq$ and $\disj$, as they are not expressible in ShEx: they can be used to enforce or forbid a diamond pattern, whereas ShEx shapes are invariant under unravelling the graph.

\begin{definition}[Restricted SHACL] \label{def:restricted-shacl}\emph{Restricted SHACL} is the syntactic fragment of SHACL defined by the grammar
\begin{align*}
\varphi
\gDef \
& \top
\gMid \bot
\gMid \hasvalue(c)
\gMid \test(\vtype)
\gMid s
\gMid \closed(Q)
\gMid\\
& \neg \varphi
\gMid \varphi \land \varphi
\gMid \varphi \lor \varphi
\gMid \geqn{n}{p}{\varphi}
\gMid \geqn{n}{p^{-}}{\varphi}
\gMid\\
& \leqn{n}{p}{\varphi}
\gMid \leqn{n}{p^{-}}{\varphi}
\gMid \exists \pathExpr.{\varphi}
\gEnd
\end{align*}
\end{definition}

Having defined the restricted fragments,  we can formulate the equivalence result. 

\begin{theorem}
    \label{thm:equivalent-shex-shacl} 
    Restricted ShEx (with $\GFP$) and restricted SHACL with $\LFP$ have the same expressive power and there are effective back-and-forth translations between them.
\end{theorem}

The proof proceeds by normalizing both formalisms to even simpler fragments and applying the duality principle
\ifArxivVersion
(see Appendix~\ref{app:equivalence-shex-shacl}).
\else 
(see \cite{arxiv}).
\fi
\subsection{Complexity}
\label{ssec:complexity}
\label{sec:complexity}

\begin{table}[t]
\setlength{\tabcolsep}{2pt}

\begin{tabular}{lllll}
\toprule
& \multicolumn{2}{l}{\textbf{Data complexity}} 
& \multicolumn{2}{l}{\textbf{Combined complexity}} 
\\ 
 & $\LFP$/$\GFP$ & $\SMS$ & $\LFP$/$\GFP$ & $\SMS$\\
\midrule
$\SSL$ &  $\ptime$ $\dagger$ & $\nptime$  $\ddagger$ & $\ptime$ $\dagger$ & $\nptime$ $\ddagger$ \\ 
SHACL  & $\ptime$ $\dagger$ & $\nptime$ $\ddagger$ & $\ptime$ $\dagger$ &  $\nptime$ $\ddagger$\\ 
ShEx  & $\ptime$ (Th.\,\ref{thm:shex-datacompl})& $\nptime$ (Th.\,\ref{thm:shex-datacompl})& $\ptime^\nptime$ (Th.\,\ref{thm:shex-combined}) & $\nptime^\nptime$ (Th.\,\ref{thm:shex-combined})  \\
\bottomrule 
\end{tabular}
\quad $\dagger$\cite{ACORSS20} \\[0.5ex]
\quad $\ddagger$\cite{CRS18}
\caption{Complexity of validation}
\label{tab:complexity}
\end{table}

An overview of prior and new results on the complexity of validation is given in Table~\ref{tab:complexity}. 
Both the combined and data complexity have been known for SHACL under \LFP and \SMS \cite{ACORSS20,CRS18}.
For $\GFP$ the complexity is the same,  as it comes from iterating a polynomial number of times the same operator as for $\LFP$. 
The upper bounds cover $\SSL$, which is a fragment of SHACL. The $\nptime$ lower bounds carry over too, because the SHACL schemas used in the proofs are readily expressible in $\SSL$, but we found a much simpler proof (for a fixed schema), by a straightforward reduction from 3COLOURABILITY 
\ifArxivVersion
(see Appendix~\ref{app:3col}).
\else 
(see \cite{arxiv}).
\fi
For ShEx it has only been known that combined complexity for the fragment without Boolean connectives is $\nptime$-complete under $\GFP$~\cite{SBG15}.
We fill this gap by establishing tight bounds for full ShEx under all three semantics.

\begin{theorem}[Data complexity of ShEx]
\label{thm:shex-datacompl}
For a fixed ShEx schema $\schema$, the problem of deciding whether a given graph $\graph$ conforms to $\schema$ is in $\ptime$ for both $\LFP$ and $\GFP$, and $\nptime$-complete for $\SMS$.
\end{theorem}

The $\nptime$ lower bound for $\SMS$ in Theorem~\ref{thm:shex-datacompl} carries over from $\SSL$ which is subsumed by ShEx. Let us move to the upper bounds.    

The validation problem for all three semantics amounts to applying the operator over shape assignments that is associated with the shape catalogue: for $\SMS$ we apply the operator to the guessed shape assignment to check if it is a fixpoint, and for $\LFP$ and $\GFP$ we iterate the operator over a suitable initial shape assignment until a fixpoint is reached, which is guaranteed to happen in a polynomial number of iterations.  
Applying the operator to a shape assignment $\alpha$ amounts to evaluating shape expressions in the catalogue, which is straightforward (following their syntactic structure) except for one subtask: checking whether the neighbourhood of a node in the graph matches a triple expression. 
Thus, in order to prove the upper bounds from Theorem~\ref{thm:shex-datacompl}, it suffices to show that this subtask is tractable. We do it in Proposition~\ref{prop:neighmatch}.   
%
%
By adding auxiliary declarations to the catalogue, we can assume that triple expressions are \emph{shallow}, in the sense that  the shape expression $\varphi$ in atomic triple expressions $p.\varphi$ and $p^-.\varphi$ must be a shape name. 

\begin{algorithm}[tb]
\small
\SetInd{0.5em}{0.5em}
\caption{Matching Triple Expressions}
\label{alg:triple-ex-test}
\KwIn{
graph $\graph$, shape assignment $\alpha$, node $v$, shallow triple expression $e=f_0 \shexeach (\shexneginv{R})^{*} $  or $ e= f_0 \shexeach (\shexneginv{R})^{*} \shexeach (\shexneg{Q})^{*}$
}
\KwOut{\textit{true} iff $\neigh^\pm_\graph(v) \in \semdelta{e}{\alpha}{\graph}$}

$M_0:=\mathcal{M}^{\alpha,e}_{\graph,v}$ \tcp*{\scriptsize extract multiset from $\neigh^\pm_\graph(v)$}

\SetKwFunction{FRecBS}{RecMatch}
\Return \FRecBS{$M_0,e$}

\SetKwProg{Fn}{Function}{:}{}
\Fn{\FRecBS{$M,f$}}{
  \If{$f=f_1\shexeach f_2$}{
        Nondeterministically split $M$ into  $M_1$ and $M_2$ \;
        \Return \FRecBS{$M_1,f_1$} $\&\&$ \FRecBS{$M_2,f_2$}\!\!\!\!\!\!\!
    }
  \If{$f=u^*$}{
        Nondeterministically split $M$ into  $M_1$ and $M_2$ \;
       \Return \FRecBS{$M_1,u$} $\&\&$ \FRecBS{$M_2,f$}\!\!
    }
  \If{$f=f_1\shexone f_2$}{
        \Return \FRecBS{$M,f_1$} $||$ \FRecBS{$M,f_2$} 
    }
  \lIf{$f=q.s$ and $M$ is the singleton of $(q,T)$ with $s\in T$}{
        \Return \textit{true}       
    }
  \lIf{$f=\lnot Q$ and $M$ is the singleton of $(p,T)$ with $p\not\in Q$}{
        \Return \textit{true}       
    }
  \lIf{$f=\lnot R^-$ and $M$ is the singleton of $(p^-,T)$ with $p\not\in R$}{
        \Return \textit{true}       
    }
  \lIf{$f=\varepsilon$ and $M$ is empty}{
        \Return \textit{true}
    }
  \Return \textit{false}    
}
\end{algorithm}

\begin{proposition}
\label{prop:neighmatch}
  Fix a catalogue $\cat$ and a shallow triple expression $e$ as in  Def.~\ref{def:shex-syntax-semantics} occurring in $\cat$. Given as input a graph $\graph$, a shape assignment $\alpha$ for $\graph$ and $\cat$, and a node $v$, we can decide in polynomial time whether  $\neigh^\pm_\graph(v)  \in \semdelta{e}{\alpha}{\graph}$. 
\end{proposition}

\begin{proof}
In this proof, $q$ stands for a predicate $p$ or an inverse predicate $\shexinverse p$. Let $\Pi_e$ be the set of predicates $p$ used in $e$ (including those in subexpressions $\lnot R^-$ and $\lnot Q$) and their inverses $p^-$, and let $\mathsf{other}$ be a fresh predicate not used in $e$.
Let $\Gamma_{e,q}$ be the set of shape names $s$ s.t. $q.s$ occurs in $e$.

An \emph{edge type} for $e$  is a pair $(q,T)$ where $q\in \Pi_e \cup \{\mathsf{other}, \mathsf{other}^-\}$ and $T\subseteq \Gamma_{e,q}$.  
Let $\Sigma_{e}$ denote the set of all edge types for $e$. In particular, 
$(\mathsf{other},\emptyset), (\mathsf{other}^-,\emptyset) \in \Sigma_{e}$. Note that $\Sigma_{e}$ might be exponential in the size of $e$, but this is fine as we consider $e$ to be fixed. 

A triple $(v,q,v')$ in $\neigh^\pm_\graph(v)$ 
with $q\in \Pi_e$ has type $(q,T)$ under shape assignment $\alpha$ if  
$v'\in \alpha(s)$ for each $s \in T$ and $v'\not\in \alpha(s)$ for each $s \in \Gamma_{e,q}\setminus T$.  A triple $(v,q,v')$ in $\neigh^\pm_\graph(v)$ 
with $q\not \in \Pi_e$ has type $(\mathsf{other}, \emptyset)$ if $q$ is a predicate, and $(\mathsf{other}^-, \emptyset)$ if $q$ is an inverse predicate. Note that each triple in $\neigh^\pm_\graph(v)$ has exactly one type.

We define $\mathcal{M}^{\alpha,e}_{\graph,v}$ as the multiset of types of triples in $\neigh^\pm_\graph(v)$. That is, $\mathcal{M}^{\alpha,e}_{\graph,v}$ is a multiset over  $\Sigma_e$ and the multiplicity of $(q,T)$ in   $\mathcal{M}^{\alpha,e}_{\graph,v}$ is equal to the number of triples $(v,q,v')$ in $\neigh^\pm_\graph(v)$ that have type $(q,T)$. 

Let $M$ be a multiset over $\Sigma_e$. For $(q,T)\in\Sigma_e$ we write $M(q,T)$ for the multiplicity of $(q,T)$ in $M$.
We call $M$  \emph{the singleton of $(q,T)$} if $M(q,T)=1$ and $M(q',T')=0$ for all $(q',T') \in \Sigma_e \setminus \{(q,T)\}$. 
We say $M$ is \emph{split} into multisets $M_1$ and $M_2$ over $\Sigma_e$ if $M_1(q,T)+M_2(q,T)=M(q,T)$ for every $(q,T)\in \Sigma_e$.

Consider the nondeterministic Algorithm~\ref{alg:triple-ex-test}, which recursively splits
$\mathcal{M}^{\alpha,e}_{\graph,v}$ while attempting to match subexpressions of $e$. 
It is easy to see that $\neigh^\pm_\graph(v) \in \semdelta{e}{\alpha}{\graph}$ iff 
\begin{equation}
\label{termination} 
\text{some run of Algorithm~\ref{alg:triple-ex-test} on  
$\graph,\alpha,v,e$ returns  \textit{true}.} \tag{$\dagger$}
\end{equation}
Due to the possible presence of $*$ in $e$, Algorithm~\ref{alg:triple-ex-test} need not terminate. Yet, condition \eqref{termination} can be effectively checked. 
In fact, it can be checked in polynomial time when $e$ is assumed to be fixed and only the size of $\graph,\alpha,v$ may grow, because the recursive subroutine $\mathsf{RecMatch}$ can be seen as an \emph{alternating} procedure that requires only logarithmic space. 
Indeed, storing the multisets $M,M_1,M_2$ used in the $\mathsf{RecMatch}$ requires only logarithmic space: for each multiset we only need  $|\Sigma_e|$ counters that count up to the number of nodes 
in $\graph$. Hence, the problem of checking \eqref{termination} is in $\alogspace$ and hence in $\ptime$ since  $\alogspace=\ptime$.
  \end{proof}

For the combined complexity the upper bounds are straightforward, and the lower bounds use reductions from suitable variants of the $\mathsf{SAT}$ problem 
\ifArxivVersion 
(see Appendix~\ref{app:shex-combined}).
\else 
(see \cite{arxiv}).
\fi

\begin{theorem}[Combined complexity of ShEx]
\label{thm:shex-combined}
The problem of deciding whether a given graph $\graph$ conforms to a given ShEx schema is
$\ptime^\nptime$-complete for $\LFP$ and $\GFP$, and 
$\nptime^\nptime$-complete for $\SMS$. 
\end{theorem}



\section{Conclusion}
\label{sec:conclusion}

The behaviour of SHACL validators on recursive schemas is unpredictable, which is only natural since its language design does not settle how to deal with recursion. Indeed, while ShEx systems consistently behave in line with \GFP, SHACL validators mostly appear to be aligned with brave \SMS, but some fail on simple schemas and others appear to mix behaviours depending on structural properties of the schema. This heterogeneity explains long-standing interoperability glitches and motivates standardizing an explicit, implementation-independent recursion contract. 

Our results show that \LFP, the hot candidate for recursion in SHACL, is a sensible choice:
\begin{enumerate}
\item It is natural. It is the semantics of common rule languages such as Datalog, the language Rel that is used in industry applications~\cite{Aref25}, as well as RDF-specific ones like SPIN~\cite{KHI11} and RIF~\cite{rif}. It is also the standard semantics for stratified programs in logic programming.
\item It allows effective translations between large fragments of ShEx and SHACL, which improves interoperability and compatibility between the two languages.
\item It avoids the NP-hardness of validation under \SMS.
\end{enumerate}

\section*{Acknowledgements}

This work was initiated during Dagstuhl Seminar 24102: \emph{Shapes in Graph
Data: Theory and Implementation}. It was funded by 
SHAKIGRA: Shaping Knowledge and Interoperable Graphs, code: PID2024-157010OB-I00 from the Spanish Research Agency, project SEK-25-GRU-GIC-24-089 from Asturias Regional research Plan and COST Action CA23147 GOBLIN (Labra Gayo);
Poland's NCN grant 2018/30/E/ST6/00042 (Murlak), the Austrian Science Fund (FWF)  projects 10.55776/COE12 and P30873 (\v{S}imkus);
and partially supported by the Province of Bolzano and FWF through project OnTeGra (DOI 10.55776/PIN8884924, Savkovi\'c).

\section*{AI Declaration}

The authors have not employed any Generative AI tools.

\bibliographystyle{kr}
\bibliography{bib/references-camera-ready}

 \appendix

 \section{Proof of Theorem~\ref{thm:equivalent-shex-shacl}}
\label{app:equivalence-shex-shacl}

The proof follows these steps. 
In Section~\ref{sec:appendix-equivalence-microshex-microshacl} we define two smaller fragments \microshex{} of ShEx and \microshacl{} of SHACL and show that they are equally expressive for the \LFP semantics for \microshacl.
Then, in Section~\ref{sec:appendix-normalise-shex} we show that restricted ShEx can be normalised to \microshex{}, and in Section~\ref{sec:appendix-normalise-shacl} that restricted SHACL can be normalised to \microshacl{}.

We say that $x_1$ and $x_2$ are equivalent, written $x_1 \equiv x_2$, if for all graphs $\graph$, catalogues $\cat$, and assignments $\alpha \subseteq \dom{\cat} \times (\nodes(\graph) \cup \const(\cat))$, it holds $\semdelta{x_1}{\alphaRel}{\graph} = \semdelta{x_2}{\alphaRel}{\graph}$.
Above, each of $x_1, x_2$ can be a ShEx or a SHACL shape, or a ShEx triple expression.

\subsection{\microshex{} and \microshacl}
\label{sec:appendix-equivalence-microshex-microshacl}

\begin{definition}[\microshex] 
\label{def:syntax-microshex}
Let $\top = \varepsilon \shexeach (\shexneginv{\emptyset})^{*} \shexeach (\shexneg{\emptyset})^{*}$. 
\microshex{} is given by the grammar
    \begin{align*}
    \varphi  \gDef & \test(c) 
    \gMid \ \test(\vtype) 
    \gMid \shexref s
    \gMid \lnot \varphi
    \gMid \varphi \land \varphi
    \gMid \varphi \lor \varphi
    \gMid \shexneigh{g}
    \gMid \shexneigh{h} 
    \gEnd \\
    g \gDef & \varepsilon \shexeach \top \ 
    \gMid p.\varphi^{n} \;\shexeach \top \ 
    \gMid \shexinverse p.\varphi^{n} \;\shexeach \top
    \gEnd \\
    h \gDef & \varepsilon \shexeach (\shexneginv{\emptyset})^{*} \ \gMid \ p. \shexneigh{\varepsilon \shexeach \top}^{*} \shexeach h
    \gEnd 
\end{align*}
where $p.\varphi^{n}$ is a syntactic shortcut for $p.\varphi\shexeach\cdots\shexeach p.\varphi$ repeated $n$ times.
\end{definition}
\microshex{} allows two kinds of limited triple expressions.
Expressions $g$ are always completely open (i.e. $(\shexneginv{R})^{*} \shexeach (\shexneg{Q})^{*}$ from Def.~\ref{def:shex-syntax-semantics} is always $\top$, obtained for $R = Q = \emptyset$) and are otherwise limited to $p.\varphi^{n}$ or $\shexinverse p.\varphi^{n}$.
Expressions $h$ are necessary to capture closed ShEx expressions and the $\closed$ construct of SHACL. 
Note that in Def.~\ref{def:syntax-microshex}, the occurrences of $\varepsilon$ are necessary to ensure that \microshex{} is a syntactic fragment of ShEx. We will omit them in the sequel, as $\varepsilon$ is a neutral element for the `$\shexeach$' operator.

\begin{definition}[\microshacl]
\label{def:syntax-microshacl}
\microshacl{} is given by the grammar
\begin{align*}
\varphi
  \gDef &
  \top
  \gMid \bot
  \gMid \hasvalue(c)
  \gMid \test(\vtype)
  \gMid s
  \gMid \closed(Q)
  \gMid \neg \varphi
  \gMid \varphi \land \varphi \gMid \\
  & \varphi \lor \varphi
  \geqn{k}{p}{\varphi}
  \gMid \geqn{k}{p^{-}}\varphi
  \gMid \leqn{k}{p}{\varphi}
  \gMid \leqn{k}{p^{-}}\varphi
  \gEnd
\end{align*}
\end{definition}
\microshacl{} forbids the $\disj$ and $\eq$ constructs, as well as paths different from a single predicate. 

\begin{proposition}
\label{prop:equivalence-microshex-microshacl}
(a) For every \microshex{} schema $\schema = (\cat, \sel)$, there exists a \microshacl{} schema $\schema' = (\cat', \sel)$ s.t. a graph $\graph$ conforms to $\schema$ iff $\graph$ \GFP-conforms to $S'$. 
Conversely, (b) for every \microshacl{} schema $\schema = (\cat, \sel)$, there exists a \microshex{} schema $\schema' = (\cat', \sel)$ s.t. a graph $\graph$ \GFP-conforms to $\schema$ iff $\graph$ conforms to $S'$.
\end{proposition}
\begin{proof}
We show how \microshex{} shapes can be translated to equivalent \microshacl{} shapes, and conversely. 
Then, for the (a) direction of the proposition, we define $\cat' = \{\decl{s}{\varphi'} \mid s \in \dom{\cat} \text{ and } \varphi' \equiv \varphi\}$ where $\varphi$ is a ShEx shape and $\varphi'$ is a SHACL shape, and conversely for the (b) direction.

The following  equivalences hold between \microshex{} shapes (left-hand side) and \microshacl{} shapes (right-hand side), assuming that the \microshex{} shape $\varphi$ is equivalent to the \microshacl{} shape $\varphi'$:
\begin{gather*}
\test(c) \equiv \test(c)\,, \quad \test(\vtype) \equiv \test(\vtype)\,, 
\quad \shexref s \equiv s\,, \quad \shexneigh{\top} \equiv \top\,,\\
\shexneigh{p_1.\shexneigh{\top}^* \shexeach \cdots \shexeach p_n.\shexneigh{\top}^{*} \shexeach (\shexneginv{\emptyset})^{*}} \equiv \closed(\{p_1, \ldots, p_n\})\,,\\
\shexneighzero{p.\varphi^k} \equiv \geqn{k}{p}{\varphi'}\,,\quad
\shexneighzero{\shexinverse p.\varphi^k} \equiv \geqn{k}{\shexinverse p}{\varphi'}\,.
\end{gather*}

These equivalences directly lead to an effective translation from \microshex{} shapes to equivalent \microshacl{} shapes, and vice-versa.
\end{proof}

Using duality of least and greatest fixed points (Prop.~\ref{prop:duality-lfp-gfp}), and the fact that $\geqn{n}{p}{\varphi}$ is the dual of $\leqn{n}{p}{\varphi}$ thanks to the equivalence $\leqn{n}{\pi}{\varphi} \equiv \neg (\geqn{n+1}{\pi}{\varphi})$, it easily follows that
\begin{corollary}
\microshacl{} under \LFP semantics is equally expressive to \microshex.
\end{corollary}

\subsection{Normalisation of Restricted ShEx}
\label{sec:appendix-normalise-shex}
%
Without loss of generality, we assume that all triple expressions are shallow. 
The set of \emph{triple constraints} of a (shallow) triple expression $e$, denoted $\ttc{e}$, is the set of $e$'s sub-expressions of the form $p.s$ or $\shexinverse p.s$ or $(\shexneg Q)$ or $(\shexneginv R)$.
We let $\preds(e) = \{p \mid p.s \in \ttc{e}\}$ be the set of predicates that appear in $e$.
For every $p \in \preds(e)$, and we let $\shapes{e, p} = \{s \mid p.s \in \ttc{e}\}$ be the shapes that appear in $e$ with predicate $p$.

\paragraph{Properties of triple expressions}
Following \cite{SBG15}, triple expressions can be seen as regular expressions with commutative concatenation over an alphabet of triple constraints, where `;' plays the role of concatenation. 
It is then easy to show that the following equivalences hold for all closed triple expressions $f_1, f_2, f_3$ and for $\mathit{op} = (\shexneginv{R})^{*} \shexeach (\shexneg{Q})^{*}$ or $\mathit{op} = (\shexneginv{R})^{*}$ for some $R, Q \subseteq \IRIs$:
\begin{gather*}
(f_1 \shexeach f_2) \shexone f_3 \equiv (f_1 \shexeach f_3) \shexone (f_2 \shexeach f_3)\,, \qquad (f_1 \shexone f_2)^{*} \equiv f_1^{*} \shexeach f_2^{*}\,,\\
\shexneigh{(f_1 \shexone f_2) \shexeach \mathit{op}} \equiv \shexneigh{f_1 \shexeach \mathit{op}} \lor \shexneigh{f_2 \shexeach \mathit{op}}.
\end{gather*}%
Additionally, every triple expression $f \shexeach (\shexneginv{R})^{*} \shexeach (\shexneg{Q})^{*}$ is equivalent to one in which $\preds(e) \subseteq R \cup Q$, and similarly for expressions without the $(\shexneg{Q})^{*}$ component, thanks to these equivalences:
\begin{align*}
f \shexeach (\shexneginv{R})^{*} \equiv\;& f \shexeach\, p^{-}.\{\top\}^{*} \shexeach\; (\shexneginv{(R\cup\{p\})})^{*}\\
f \shexeach (\shexneginv{R})^{*} \shexeach (\shexneg{Q})^{*} \equiv\;& f \shexeach\, p.\{\top\}^{*} \shexeach\; (\shexneginv{R})^{*} \shexeach (\shexneg{(Q \cup \{p\})})^{*}
\end{align*}

\subsubsection{Normalisation}
\label{sec:appendix-deterministic-triple-expr}
In the sequel we show how Restricted ShEx shapes can be translated to \microshex{} shapes.

A triple expression $e$ is called \emph{deterministic} if for all assignments $\alpha$, graphs $\graph$, catalogues $\cat$, and triples $(v, p, u) \in \graph$, $e$ contains at most one triple constraint $t$ s.t. $\{(v, p, u)\} \in \semdelta{t}{\alphaRel}{\graph}$, and contains at most one inverse triple constraint $t$ s.t. $\{(v, \shexinverse p, u)\} \in \semdelta{t}{\alphaRel}{\graph}$.
Every triple expression $e$ can be effectively transformed into an equivalent deterministic one $\dettc(e)$ as follows. 
%
Let $e = f \shexeach (\shexneginv{R})^{*} \shexeach (\shexneg{Q})^{*}$; the case where $e = f \shexeach (\shexneginv{R})^{*}$ is similar.
W.l.o.g. we assume that $\preds(e) \subseteq R \cup Q$.
\illustration{If $e = \left(p.s_1 \shexeach q.s_2 \shexone p.s_3\right) \,\shexeach\, (\shexneg{\{p,r\}})^{*}$. That is, $f = p.s_1 \shexeach q.s_2 \shexone p.s_3$. 
Then $\preds(e) = \{p, q\}$ and \\
$e \equiv \left(p.s_1 \shexeach q.s_2 \shexone p.s_3\right) \;\shexeach\; q.\shexneigh{\top}^{*} \;\shexeach\; (\shexneg{\{p, r, q\}})^{*}$.}
%
%
\illustration{
With $e$ as above, $\shapes{e,p} = \{s_1, s_3\}$ and $\shapes{e,q} = \{s_2\}$.}
For every $p \in \preds(e)$ and every $X \subseteq \shapes{e, p}$, define the shape $\Phi_{e,p,X}$: 
$$
\Phi_{e,p,X} = \bigwedge_{s \in X} s \land \bigwedge_{s \in \shapes{e,p} \setminus X} \neg s
$$
\illustration{With $e$ as above, there are four different shapes $\Phi_{e,p,X}$:\\
$\Phi_{e,p, \{s_1\}} = s_1 \land \neg s_3$\\
$\Phi_{e,p, \{s_3\}} = \neg s_1 \land s_3$\\
$\Phi_{e,p,\{s_1,s_3\}} = s_1 \land s_3$\\
$\Phi_{e,p, \emptyset} = \negs_1 \land \negs_3$}
Now, for every triple constraint $p.s$ of $e$, we define
$$
\dettc(e, p.s) = \mathop{|}_{X \subseteq \shapes{e,p} \mid s \in X} p.\Phi_{e,p,X}.
%
$$
\illustration{With $e$ as above,\\
$\dettc(e, p.s_1) = \Phi_{e,p, \{s_1 \land \negs_3\}} \shexone \Phi_{e,p, \{s_1, s_3\}}$\\
$\dettc(e, p.s_3) = \Phi_{e,p, \{\negs_1 \land s_3\}} \shexone \Phi_{e,p, \{s_1, s_3\}}$}
Then $\dettc(e)$ is obtained by replacing in $e$ every triple constraint $p.s$ by $\dettc(e, p.s)$.
The fact that $\dettc(e)$ is deterministic follows from $\preds(e) \subseteq R \cup Q$ and from the definition of the $\Phi_{e,p,X}$. 
Showing that $e \equiv \dettc(e)$ is not hard using the parallel between triple expressions and regular expressions with commutative concatenation.

\newcommand{\tc}[1]{T_{#1}}

\medskip
Assume a fixed Restricted ShEx catalogue $\cat$ and a fixed graph $\graph$.
For every Restricted ShEx shape, we construct an equivalent \microshex{} shape.
The only non-trivial case is for shapes of the form $\shexneigh{e'}$.
We start by computing $\dettc(e')$. 
Note that if $e'$ is a Restricted ShEx triple expression, then so is $\dettc(e')$.
In the sequel, let $\dettc(e') = e = f \shexeach \mathit{op}$ with $\mathit{op} = (\shexneginv{R})^{*} \shexeach (\shexneg{Q})^{*}$ for some $R,Q \subseteq \IRIs$.
The case without $(\shexneg{Q})^{*}$ is similar.

Using the properties of triple expressions mentioned earlier and the fact that $e$ is Restricted ShEx, we can push the `$\shexeach$' and `$*$' operators under the `$\shexone$', then transform the top-level `$\shexone$' into disjunction.
Thus, 
\begin{equation}
    \label{eq:shex-shape-to-disjunctive-normal-form}
    \shexneigh{e} \equiv \shexneigh{F_1 \shexeach \mathit{op}} \lor \cdots \lor \shexneigh{F_{n'} \shexeach \mathit{op}}
\end{equation}
with $n' \ge 1$ and where every $F_j$ is a `$\shexeach$' combination of triple constraints of the form
\begin{equation}
\label{eq:shex-disjunct-as-eachof-of-atoms}
    p_1.\Phi_1 \shexeach \cdots \shexeach p_m.\Phi_m \;\;\shexeach\;\; p_{m+1}.\Phi_{m+1}^{*} \shexeach \cdots \shexeach p_k.\Phi_k^{*},
\end{equation}
for $0 \le m \le k$ ($k=0$ means that the sequence is empty, thus reduces to $\varepsilon$).
Let $F$ be the triple expression (\ref{eq:shex-disjunct-as-eachof-of-atoms}).
For every $p.\Phi$ triple constraint of $F$, let $\nbtc{F, p.\Phi}$ be the number of times $p.\Phi$ is repeated in $F$ without a star (i.e. we do not count occurrences of $p.\Phi^{*}$).

\begin{lemma}
\label{lem:each-of-to-conjunction}
    Let $e = f \shexeach (\shexneginv{R})^{*} \shexeach (\shexneg{Q})^{*}$ be a deterministic triple expression with (for some $0 \le m \le k$), where $f$ is the expression (\ref{eq:shex-disjunct-as-eachof-of-atoms}) above.
    Then $\shexneigh{e}$ is equivalent to the conjunction shape
\begin{align}
    \label{eq:at-least-k}
    \Psi = &&& \bigwedge_{p.\Phi \in \ttc{f}} \shexneighzero{p.\Phi^{\nbtc{f, p.\Phi}}} \\
    \label{eq:at-most-k}
    \land &&& \bigwedge_{p.\Phi \in \ttc{f} \text{ s.t. }p.\Phi^{*}\not\in\ttc{f}} \neg \shexneighzero{p.\Phi^{1+\nbtc{f,p.\Phi}}}\\
    \label{eq:no-not-phi}
    \land &&& \bigwedge_{p \in \preds(f)} \neg \shexneighzero{p. \Phi_{f,p,\emptyset}} \\
    \label{eq:no-R}
    \land &&& \bigwedge_{p \in R \setminus \preds(f)} \neg \shexneighzero{\shexinverse p.\shexneigh{\top}} \\
    \label{eq:no-Q}
    \land &&& \bigwedge_{p \in Q \setminus \preds(f)} \neg \shexneighzero{p.\shexneigh{\top}}
\end{align}
\end{lemma}
\begin{proof}
%
Let $v$ be a node in $\graph$, and let $N'$ be its neighbourhood.
Because $e$ is deterministic, it is not hard to see that if $N' \in \semdelta{e}{\alpha}{\graph,v}$, then $N'$ is a disjoint union $N' = \biguplus_{T \in \ttc{f}} E_T$ where, for every $T \in \ttc{f}$, $E_T$ is the set of triples from $N'$ that match the triple constraint $T$.
Additionally, $|E_T| = \nbtc{e, T}$ if $T^{*}$ does not appear in $e$, and $|E_T| \ge \nbtc{e, T}$ otherwise.

Assume now that $N' \in \semdelta{\Psi}{\alpha}{\graph}$, and for every $T \in \ttc{e}$, let $E'_T \subseteq N'$ be the set of triples from $N'$ that match the triple constraint $T$, and let $E'_0 \subseteq N'$ be the set of triples from $N'$ that match none of the triple constraints of $e$.
Because $e$ is deterministic, we deduce that $N' = E_0 \uplus \biguplus_{T \in \ttc{e}} E'_T$.
The conjuncts (\ref{eq:no-not-phi}), (\ref{eq:no-R}) and (\ref{eq:no-Q}) guarantee that $E'_0 = \emptyset$.
The conjuncts (\ref{eq:at-least-k}) and (\ref{eq:no-not-phi}) guarantee that $|E'_T| \ge \nbtc{e, T}$, and together with the conjuncts (\ref{eq:at-most-k}) they guarantee that if $T$ is not starred in $e$, then $|E'_T| = \nbtc{e, T}$.
From these observations it easily follows that $N' \in \semdelta{e}{\alpha}{\graph}$ iff $N' \in \semdelta{\Psi}{\alpha}{\graph}$.
\end{proof}

Note that the formula $\Psi$ in Lemma~\ref{lem:each-of-to-conjunction} is \microshex.
Therefore, we have shown 
\begin{proposition}
    Every Restricted ShEx schema is equivalent to a \microshex{} schema.
\end{proposition}

\subsection{Normalisation of Restricted SHACL}
\label{sec:appendix-normalise-shacl}
In this section, we relate Def.~\ref{def:restricted-shacl} and Def.~\ref{def:syntax-microshacl}. Specifically, we argue that, under the \LFP semantics, complex path expressions as allowed in Def.~\ref{def:restricted-shacl} can be eliminated by adding additional shape declarations, thus obtaining a schema matching Def.~\ref{def:syntax-microshacl}. See Proposition 7 in \cite{OudshoornOrtizSimkus2024} for a similar approach. 

\begin{proposition}
    Under the \LFP semantics, every Restricted SHACL schema is equivalent to a \microshacl{} schema.
\end{proposition}

\begin{proof}[Proof (sketch)]
Let $(\cat, \sel)$ be a Restricted SHACL Schema. 
  We only need to show how to eliminate expressions of the form
  $\exists \pi.\varphi$. 
  We perform the following steps for all 
  $\exists \pi.\varphi$ that appear in~$\cat$.
  \begin{enumerate}
  \item Take two fresh shape names $s_{\exists \pi.\varphi}$ and $s_{\varphi}$.
  \item Add the shape declaration $\decl{s_{\varphi}}{\varphi}$ to $\cat$.
  \item Replace each occurrence of $\exists \pi.\varphi$ in $\cat$ with $s_{\exists \pi.\varphi}$.
  \item Construct a non-deterministic finite automaton $A$ that accepts the language of $\pi$. For every state $q$ of $A$,  take a fresh shape name $s_q$ and add the following declarations to $\cat$:
    \begin{enumerate}
    \item $\decl{s_{\exists \pi.\varphi}}{s_{q_0} \lor \cdots \lor s_{q_n}}$ where $q_0, \ldots, q_n$ are the initial states of $A$, 
    \item $\decl{s_q}{s_{\varphi}}$ for each accepting state $q$ of $A$, and
    \item $\decl{s_q}{\exists r.s_{q'}}$ for each transition $(q,r,q')$ of $A$.
    \end{enumerate} 
  \end{enumerate}
  The resulting set of shape declarations is denoted $\cat'$. 
  It is straightforward to check that $(\cat', \sel)$ is a \microshacl{} schema equivalent to $(\cat, \sel)$.  We can also note that the construction of
  $\cat'$ takes only polynomial time in the size of
  $\cat$.
\end{proof}


 \section{NP-Hardness of \SMS for stratified \SSL}
\label{app:3col}

We provide an alternative proof of NP-hardness of validation under \SMS for a fixed stratified $\SSL$ schema (and thus also for a fixed stratified SHACL schema). We reduce from 3COLOURABILITY. 

For a directed graph
$H$, we construct an RDF graph $\mathcal{G}_{H}$ and a stratified $\SSL$ catalogue $\cat$ (independent from $H$) such that $H$ is 3-colourable iff $\mathcal{G}_{H}$ \SMS-conforms to the schema $(\cat, \{
\mathit{Ok}\colon \hasvalue(s)\})$. 

We first describe $\mathcal{G}_H$. We use predicate $\mathsf{edge}$ to represent the edges of $H$: for every edge $(u,v)$ of $H$, we add to
$\mathcal{G}_{H}$ the triple $(u,\mathsf{edge},v)$. We use predicate $\mathsf{self}$ for self-loops, which are used to generate arbitrary colourings of $H$: for each node $v$ of $H$, we add to $\mathcal{G}_{I}$ the
triple $(v,\mathsf{self},v)$.
Finally, we add a fresh apex node $s$ that serves as a spy-point, observing all nodes of $H$ via 
predicate $\mathsf{spy}$:
for each node $v$ of $H$, we add to $\mathcal{G}_{H}$ the
triple $(s,\mathsf{spy},v)$.

The catalogue $\mathcal{C}$ contains the following declarations:
\begin{align*}
  \mathit{Colour}_i &: \exists \mathsf{self}. \mathit{Colour}_i\,,\quad 1 \leq i \leq 3 \\
  \mathit{Error}    &: \bigwedge_{i=1}^3\neg \mathit{Colour}_i \;\lor\;
                       \bigvee_{i=1}^3 \mathit{Colour}_i\land \exists \mathsf{edge}.\mathit{Colour}_i \\
  \mathit{Ok}       &: \neg \exists \mathsf{spy}. \mathit{Error}\,.
\end{align*}

To see that the reduction is correct, observe that in a correct shape assignment, shape name $\mathit{Colour}_i$ can be assigned to an arbitrary subset of nodes of $H$, which is how we guess a colouring of $H$. An error is detected in a node if it gets no colour or it has the same colour as one of its neighbours. Finally, $\mathit{Ok}$ can be assigned to the spy point $s$ iff  no error is detected in any of the spied nodes.

 \section{Proof of Theorem~\ref{thm:shex-combined}}
\label{app:shex-combined}

We begin with $\SMS$. For the $\nptime^\nptime$ upper bound we guess a shape assignment and verify that it is correct and witnesses conformance. These checks can be done easily in $\ptime$, when given access to an oracle that answers if the neighbourhood of a given node matches a given triple expression. This we can assume in $\nptime^\nptime$, because the subproblem can be solved in $\nptime$ by guessing an ordering of edges in the neighbourhood into a sequence, and feeding it to the nondeterministic automaton constructed (in $\ptime$) from the triple expression.

For the lower bound we use the following variant of the  standard $\nptime^\nptime$-complete  problem, $\mathsf{\exists\forall SAT}$: Given an CNF propositional formula over variables $\bar x$ and $\bar y$, decide if there is a valuation of $\bar x$ such that for each valuation of $\bar y$ the formula is false. 

Let us fix a CNF formula with clauses $c_1, c_2, \dots, c_n$ over  variables $x_1, x_2, \dots, x_k$ and $y_1, y_2, \dots, y_\ell$. We will be using $c_1, c_2, \dots, c_n$ and $x_1, x_2, \dots, x_k$ as predicates (edge labels). 

The input graph is a $(k+n)$-pointed star with self-loops: it consists of a central node $a$ 
with an outgoing $c_i$-edge for $i=1, 2, \dots, n$ and an outgoing $x_i$-edge for $i=1, 2, \dots, k$;  the target of each $x_i$-edge has a self-loop with label $p$.

The catalogue consists of two declarations, $s \colon \{p.s \shexeach \top\}$  and 
\begin{gather*}
t\colon \lnot 
\big\{\big(x_1.s\shexeach e_{x_1} \shexone \,  x_1.\lnot s \shexeach e_{\lnot x_1}\big) 
\shexeach \ldots \shexeach
\big(x_k.s\shexeach e_{x_k} \shexone\, x_k.\lnot s \shexeach e_{\lnot x_k}\big) 
\shexeach \\
\big(e_{y_1} \shexone \,e_{\lnot y_1}\big) 
\shexeach \ldots \shexeach
\big(e_{y_\ell} \shexone \,e_{\lnot y_\ell}\big)
\big\} 
\end{gather*}
where the triple expression $e_L$ for literal $L$ is obtained by combining with `$\shexeach$' all expressions $c_i.\{\top\}^*$ such that  $c_i$ contains literal $L$.

The schema has a single selector $t \colon \hasvalue(a)$.

In a correct shape assignment, the shape name $s$ is associated with any subset of $x_i$-successors of $a$, thereby fixing a valuation of $x_1, x_2, \dots, x_k$. 
In order to check if the neighbourhood of $a$ matches the triple expression $e$ in the declaration of $t$, we guess the left or right argument of $\shexone$ in each parenthesis. This corresponds to guessing a valuation of all variables. However, once a correct assignment $\alpha$ is fixed, in expressions of the form $(x_i.s\shexeach e_{x_i} \shexone \,  x_i.\lnot s \shexeach e_{\lnot x_i})$ we have to follow the choice dictated by $\alpha$, and we have actual choice only for variables $y_1, y_2, \dots, y_\ell$. Once all choices are made, the neighbourhood matches the resulting triple expression $e'$ if $e'$ mentions each $c_i.\{\top\}^*$ at least once (we interpret all but one occurrence as $\varepsilon$). Hence, the graph is valid if there exists a valuation of $x_1, x_2, \dots, x_k$ such that for all valuations of $y_1, y_2, \dots, y_\ell$, at least one $c_i.\{\top\}^*$ is not mentioned in the triple expression $e'$ induced by the valuations. By the definition of $e_L$, the latter holds iff the formula is false.

\smallskip 

Let us move to $\LFP$ and $\GFP$. For the $\ptime^\nptime$ upper bound we iterate the suitable operator as explained in the proof of Theorem~\ref{thm:shex-datacompl}, except that we use an $\nptime$ oracle  for testing if the neighbourhood of a given node matches a given triple expression.

For the lower bound we use the following $\ptime^\nptime$-complete problem: 
given a satisfiable 3CNF formula $\varphi$ over ordered variables $x_1,\ldots,x_n$, decide if $x_n = 1$ in the \emph{lexicographically maximal} valuation that satisfies $\varphi$. In other words, we need to check if 1 is the last bit of the maximal $n$-bit number that encodes a valuation satisfying $\varphi$. The reduction is inspired in the $\ptime^\nptime$-hardness proof for stratified logic programs under bounded predicates arities~\cite{DBLP:journals/amai/EiterFFW07}. 

Let $\varphi$ be a satisfiable 3CNF formula with clauses $c_1, \dots,c_k$ over ordered variables $x_1,\ldots,x_n$.
We shall construct an instance of the validation problem that will simulate an algorithm that solves our $\ptime^\nptime$-complete problem by iteratively computing values of variables in the lexicographically maximal valuation satisfying $\varphi$ using satisfiability tests. In iteration 1 we check if $\varphi$ is satisfiable with $x_1=1$, and store the outcome in  $t_{1}$. In iteration $i$ we check if $\varphi$ is satisfiable with $x_1 = t_1,\ \ldots,\ x_{i-1} = t_{i-1},\ x_i=1$, and store the outcome in~$t_i$. After iteration $n$, $t_{1},\dots, t_{n}$ store  the lexicographically maximal valuation satisfying $\varphi$. It suffices to check $t_n$.

The reduction is similar to the one for $\SMS$, but this time there are no variables $y_1, \dots,y_m$, and the values of $x_1, \dots, x_n$ are iteratively determined, rather than guessed. We use predicates $c_1, \dots, c_k$ and $x_1, \dots, x_n$, like before, and shape names $t_1, \dots, t_n$ and $s_1, \dots, s_n$, corresponding to subsequent iterations.
The shape catalogue contains the following  declarations, for $i=1, 2, \dots, n$: 
{\small
\begin{align*}
t_i \colon  
\big\{&
\big(x_1.s_i\shexeach e_{x_1} \shexone  x_1.\lnot s_i \shexeach e_{\lnot x_1}\big) 
\shexeach \\
&\ldots \shexeach
\big(x_{i-1}.s_{i}\shexeach e_{x_{i-1}} \shexone x_{i-1}.\lnot s_{i} \shexeach e_{\lnot x_{i-1}}\big) 
\shexeach \\
& \big(x_{i}.\{\top\}\shexeach e_{x_{i}} \big) \shexeach\\
& x_{i+1}.\{\top\}\shexeach \big (e_{x_{i+1}} \shexone \,  e_{\lnot x_{i+1}}\big) 
\shexeach \ldots \shexeach
x_{n}.\{\top\}\shexeach \big(e_{x_{n}} \shexone\, e_{\lnot x_{n}}\big) 
\big\} \\
s_i \colon \big\{ & x_1^-.t_1 \shexone \dots \shexone x_{i-1}^-.t_{i-1}\big\}
\end{align*}
}
Note that the catalogue is  non-recursive, so the $\LFP$ and $\GFP$ semantics coincide (Proposition~\ref{prop:nonrec-coincide}).
%
The graph is a $(k+n)$-pointed star: it consists of a central node $a$ with an  outgoing $x_i$-edge for $i=1, 2, \dots n$ and an outgoing $c_i$-edge for  $i=1, 2, \dots k$. The schema has a single selector $t_n \colon \hasvalue(a)$.
One can verify that $\varphi$ is a positive instance of the problem iff the graph $\LFP/\GFP$-conforms to the schema.

\filip{Ultimately, this should probably make its way to the body (Section 2), one way or another.}

\begin{proposition} \label{prop:nonrec-coincide}
For non-recursive catalogues \LFP, \GFP, b\SMS, and c\SMS semantics coincide. 
\end{proposition}
\begin{proof}[Proof sketch]
For such schemas, the initial shape assignment is effectively irrelevant, since you can rewrite the schema into an equivalent one where no shape depends on another. This is done by substituting shape names used in shapes with their declared definitions. In consequence, there is only one correct shape assignment. 
\end{proof}

\end{document}